\documentclass[aps,twocolumn,nofootinbib,rmp]{revtex4}
\usepackage{amsfonts,amssymb,amsmath}            
\usepackage[]{graphics,graphicx,epsfig}            
\usepackage{amsthm}
\pdfoutput=1
\def\identity{\leavevmode\hbox{\small1\kern-3.8pt\normalsize1}}

\newtheorem{lemma}{Lemma}
\newcommand{\ket}[1]{\left | #1 \right\rangle}
\newcommand{\bra}[1]{\left \langle #1 \right |}
\newcommand{\half}{\mbox{$\textstyle \frac{1}{2}$}}
\newcommand{\smallfrac}[2][1]{\mbox{$\textstyle \frac{#1}{#2}$}}
\newcommand{\Tr}{\text{Tr}}
\newcommand{\braket}[2]{\left\langle #1|#2\right\rangle}
\newcommand{\proj}[1]{\ket{#1}\bra{#1}}
\renewcommand{\epsilon}{\varepsilon}

\bibpunct{[}{]}{,}{n}{}{;}

\begin{document}

\title{A Review of Perfect, Efficient, State Transfer and its Application as a Constructive Tool}
\date{\today}

\author{Alastair \surname{Kay}}
\affiliation{Max-Planck-Institut f\"ur Quantenoptik, Hans-Kopfermann-Str.\ 1,
D-85748 Garching, Germany}
\affiliation{Centre for Quantum Computation,
             DAMTP,
             University of Cambridge,
             Cambridge CB3 0WA, UK}
\affiliation{Centre for Quantum Technologies, National University of Singapore, 3 Science Drive 2, Singapore 117543}
\begin{abstract}
We review the subject of perfect state transfer; how one designs the (fixed) interactions of a chain of spins so that a quantum state, initially inserted on one end of the chain, is perfectly transferred to the opposite end in a fixed time. The perfect state transfer systems are then used as a constructive tool to design Hamiltonian implementations of other primitive protocols such as entanglement generation and signal amplification in measurements, before showing that, in fact, universal quantum computation can be implemented in this way.
\end{abstract}

\maketitle

\section{Introduction}

The study of quantum mechanical systems is no longer restricted to the characterization of the properties of naturally occurring systems. Rather, with the advent of Quantum Information, the manipulation and engineering of these systems to suit our purposes has come to the fore. One of the simplest ways to understand these manipulations is to consider not universal Quantum Computation, but simpler sub-protocols. The sub-protocol of state transfer was originally proposed by Bose \cite{Bos03}, the motivation being that in a quantum computer based on, for instance, a solid state architecture, interactions are typically local, but we want to apply entangling gates between distant qubits. While a sequence of {\sc swap} gates suffices to bring those distant qubits together, this is potentially prone to massive control errors, and it would be desirable to remove the need for such stringent control. One promising alternative is to integrate a `flying qubit' such as a photon within the same system in order to transfer qubits quickly and easily. This, however, requires the successful implementation of two information processing realisations rather than just one, which is already a daunting task. Instead, the idea is to achieve state transfer by providing a pre-fabricated unit, made from the same solid state technology but with severely limited capabilities (thereby easing the fabrication) which takes as input a state in one location, and outputs it at another. In principle, since it is not necessary to interact with the device except at the input and output, it can, to some extent, be more isolated from the environment and less susceptible to decoherence. Although this was the original motivation, whether these schemes are ever practically realized is largely irrelevant; they have already provided a huge degree of insight into engineering more complex quantum protocols.

In this paper, we provide a review of how quantum systems can be engineered precisely for the task of quantum state transfer with, in some sense, minimalistic properties; there should be no interaction with the system except at initialization and read-out, and the system Hamiltonian should remain fixed in time. We start, in Sec.~\ref{sec:2}, by proving necessary and sufficient conditions for perfect state transfer in a one dimensional system, and showing that this is equivalent to the protocol of entanglement distribution. These conditions provide an infinite family of solutions, and we discuss how to design systems with specific properties, and demonstrate that one particular solution is optimal with respect to a variety of parameters. The understanding of this basic mechanism can then be applied to a variety of different systems, such as harmonic oscillators, and allow for the possibility of long-range couplings.

In Sec.~\ref{sec:highexcit}, we consider what happens if more than a single excitation is present, enabling us to generate entanglement, or to encode such that initialisation of the state of the system is unnecessary. In Sec.~\ref{sec:4}, we give a brief treatment of errors, although, if used as a constructive technique, this section is largely irrelevant. Sec.~\ref{sec:highD} describes the only known class of solutions for perfect state transfer in geometries beyond a one-dimensional chain such that the distances of transfer drops off no faster that a polynomial of the number of system qubits. Finally, in Sec.~\ref{sec:6}, we show how perfect state transfer schemes can be be modified in order to generate different types of multipartite entanglement, such as GHZ or W-states, followed by the much stronger result of how to design a Hamiltonian to perfectly implement a quantum computation without any external control.

\section{Perfect Transfer in the Single Excitation Subspace} \label{sec:2}

The study of state transfer was initiated by Bose \cite{Bos03}, who considered a 1D chain of $N$ qubits (open boundary conditions) coupled by the time-independent Hamiltonian ($X$, $Y$ and $Z$ are the standard Pauli matrices)
$$
H_{\text{Bose}}=J\sum_{n=1}^{N-1}(X_nX_{n+1}+Y_nY_{n+1}+Z_nZ_{n+1}).
$$
His protocol started with a state $\ket{\psi}\ket{0}^{\otimes N-1}$, and simply allowed the Hamiltonian to evolve the system for some transfer time $t_0$, in the hope of creating the output state $\ket{0}^{\otimes N-1}\ket{\psi}$, i.e.~the unknown state $\ket{\psi}$ has transferred from one end of the chain to the other. For $N=2,3$, transfer can be achieved perfectly. For larger $N$, it can be shown that perfect transfer is impossible \cite{Kay:2004c}. It was subsequently observed that by modulating the couplings along the length of the chain, perfect state transfer can be achieved for all $N$ \cite{Christandl,Kay:2004c,lambrop}, and it is this technique designing fixed Hamiltonians requiring no further interaction that we review here, i.e.~although a vast array of techniques for arbitrarily accurate state transfer exist with a variety of control paradigms, we restrict purely to the case of perfect transfer with no active control. While we will primarily restrict to a chain of qubits, many of the results can be generalized to larger local Hilbert space dimensions \cite{li}, including harmonic oscillators \cite{plenio:0}.

The solutions provided here for perfect state transfer represent an extremal point in the space of possible solutions since they require a minimum of interaction with the system. As some level of control is introduced, perhaps through an encoding procedure (a Gaussian wavepacket \cite{osborne04} or across multiple chains \cite{Burgarth:2}), or by some temporal modulation of the system Hamiltonian \cite{haselgrove04,networks,dark}, this can trade against the engineering requirements imposed here, and it should come as no surprise that a plethora of protocols for high quality (possibly non-deterministic) state transfer protocols have arisen. A review of some of these can be found in \cite{Bose_review}. 

While there are also extensions of these results to topologies other than one dimensional chains, such as \cite{Kay:2004c,simone3,facer,simone2,feder,topology,ahmadi}, these are all subject to the significant criticisms that the transfer distance is no better than $O(\log N)$, and they require very high coordination number. These facts suggest that they are unlikely to be realistic in a physical device, which should perhaps be restricted to localized couplings on a lattice of between 1 and 3 spatial dimensions, and very cumbersome for application in other settings. For this reason, we will henceforth restrict mainly to chains, and only mention briefly a single example\footnote{the only construction that we are aware of which gives perfect transfer in a non-1D network with reasonable transfer distance} of a more general network (Sec.~\ref{sec:highD}). Nevertheless, it is important to note that there are potentially good reasons for studying these networks. Firstly, they can typically be constructed with fixed (i.e.~not spatially varying) coupling strengths, which could reduce engineering requirements, albeit at the cost of other engineering requirements. Secondly, these more general networks admit the possibility of being able to choose which of several output nodes a state should be routed to, either when the network is first constructed \cite{facer} or, potentially, during the transfer\footnote{The routing problem can be solved on a regular lattice topology in any number of spatial dimensions with only a very mild relaxation of assumptions, and hence has a large transfer distance \cite{in_prep}.}. Finally, there is the motivation from classical network theory, that some networks, such as those based on circulant graphs, are extremely robust to imperfections. These results, however, have not yet been transferred to the quantum regime.

\begin{figure}
\begin{center}
\includegraphics[width=0.45\textwidth]{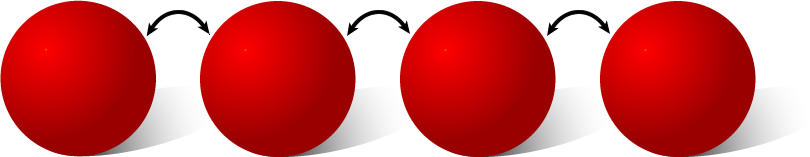}
\end{center}
\vspace{-0.5cm}
\caption{An array of qubits coupled by a nearest-neighbor interaction.} \label{fig:basica}
\vspace{-0.3cm}
\end{figure}

Thus, concentrating on the most direct method of transfer, we consider a chain of qubits coupled with the nearest-neighbor Hamiltonian
$$
H=\half\sum_{n=1}^{N-1}J_n(X_nX_{n+1}+Y_nY_{n+1})-\sum_{n=1}^NB_nZ_n,
$$
which is known as the $XX$ Hamiltonian. The coefficients $J_n$ and $B_n$ are real numbers. The analysis of state transfer is vastly simplified by the observation that
$$
\left[H,\sum_{n=1}^NZ_n\right]=0,
$$
which means that $H$ divides into a series of subspaces, characterized by the number of qubits that are in the $\ket{1}$ state. This is referred to as the number of excitations. It is not necessary to restrict to a Hamiltonian with this subspace form \cite{paternostro:1,kay-2006b}; some work has been done on Hamiltonians such as
$$
H'=\sum_{n=1}^{N-1}J_nZ_nZ_{n+1}+\sum_{n=1}^NB_nX_n,
$$
where perfect state transfer can be achieved  \cite{paternostro:1}, and many similar properties to those which we will demonstrate can be recovered \cite{paternostro:2}. We will see how these cases are also handled by the same formalism in Sec.~\ref{sec:bogoliubov}.

Most of the state transfer properties can be derived by considering the subspaces containing no more than 1 excitation and, in this case there is a large variety of other models such as Heisenberg Hamiltonians and harmonic oscillators to which these results also apply. We will explore some of these extensions in Sec.~\ref{sec:harm}.

We shall initially assume that the system of spins is prepared in the state $\ket{\underline{0}}=\ket{0}^{\otimes N}$, and that an unknown state $\ket{\psi}$ can be introduced onto spin 1 perfectly. The first of these assumptions will be relaxed in Sec.~\ref{sec:sysinit}, although the second one is necessary in a fixed Hamiltonian setting. In order to relax that constraint, one could consider time-varying coupling strengths $J_{1}$ and $J_{N-1}$ \cite{haselgrove04,Kay:08b, bruder}. By expanding $\ket{\psi}=\alpha\ket{0}+\beta\ket{1}$, we observe that the component of the initial state $\alpha\ket{0}^{\otimes N}$ is just required to evolve to the final state $\alpha\ket{0}^{\otimes N}$, but since $\ket{0}^{\otimes N}$ is an eigenstate of $H$, this happens automatically (up to a phase factor). Consequently, it suffices to consider the task $\ket{1}\ket{0}^{\otimes N-1}\rightarrow\ket{0}^{\otimes N-1}\ket{1}$, which allows the restriction to the single excitation subspace of $H$. This is described by an $N\times N$ matrix, with the basis elements of this subspace being given by
$$
\ket{\underline{n}}=\ket{0}^{\otimes n-1}\ket{1}\ket{0}^{\otimes N-n}.
$$
Using this basis, the Hamiltonian is represented by the matrix
$$
H_1=\left(
\begin{array}{cccccc}
B_1 & J_1 & 0 & \ldots & 0 & 0    \\
J_1 & B_2 & J_2 & \ldots & 0 & 0  \\
0 & J_2 & B_3 & \ldots & 0 & 0 \\
\vdots & \vdots & \vdots & \ddots & \vdots & \vdots \\
0 & 0 & 0 & \ldots & B_{N-1} & J_{N-1}    \\
0 & 0 & 0 & \ldots & J_{N-1} & B_N
\end{array}
\right),
$$
where we have neglected a term $-\sum_{n=1}^NB_n\identity$, as it evidently only affects the relative phase between the $\ket{0}$ and $\ket{1}$ components of $\ket{\psi}$ after the transfer, and this can always be corrected by using a unitary rotation.

\subsection{Entanglement Distribution}

While we will be discussing the ability to transfer an unknown state from one end of a chain to the other, it is worth being aware that this is equivalent to the task of entanglement distribution \cite{Kay:2004c}. Here, we imagine that two parties, Alice and Bob, want to share a maximally entangled state and they each control one spin of a coupled spin network.
\begin{lemma} The protocol of entanglement distribution, wherein Alice prepares a maximally entangled state, and puts half of it on the spin network, so that Bob can remove it at some later time $t_0$, such that they perfectly share a maximally entangled state, succeeds if and only if the spin network is a perfect state transfer network from Alice to Bob.
\end{lemma}
\begin{proof}
Let us assume that Alice controls two spins, 0 and 1, where 1 is the first spin of the network. Bob controls spin $B$ on the network. Evidently, if Alice prepares a state
\begin{equation}
(\ket{00}_{0,1}+\ket{11}_{0,1})\ket{0}^{\otimes N-1}/\sqrt{2}	\label{eqn:start}
\end{equation}
and waits a time $t_0$, then if the network is a perfect state transfer network, it evolves to
$$
(\ket{00}_{0,B}+\ket{11}_{0,B})\ket{0}^{\otimes N-1}/\sqrt{2},
$$
and the protocol is successful. On the other hand, consider an arbitrary network, acting on the state $(\ket{00}_{0,1}+\ket{11}_{0,1})\ket{\Phi_{in}}/\sqrt{2}$. If the output state $\ket{\Psi}$ is to be maximally entangled between spins 0 and $B$, it must satisfy
$$
U_0\otimes U_B\ket{\Psi}=(\ket{00}_{0,B}+\ket{11}_{0,B})\ket{\Phi_{out}}/\sqrt{2}
$$
for some single-qubit unitaries $U_0$ and $U_B$.
However, since $(U\otimes U^*)(\ket{00}+\ket{11})=(\ket{00}+\ket{11})$, we can take $U_0=\identity$ without loss of generality. Therefore, $\ket{0}_1$ must have transferred to $U_B^\dagger\ket{0}$, and $\ket{1}_1$ must have transferred to $U_B^\dagger\ket{1}$. Thus, by linearity, we conclude that the system would have to transfer any initial state $(\alpha\ket{0}_1+\beta\ket{1}_1)\ket{\Phi_{in}}$ to the state $U_B^\dagger(\alpha\ket{0}_B+\beta\ket{1}_B)\ket{\Phi_{out}}$ i.e.~up to a correctable local unitary on Bob's spin, the network must perform as a perfect transfer network, at least for some initial configuration $\ket{\Phi_{in}}$ of the system.
\end{proof}

\subsection{The Symmetry Matching Condition} \label{sec:symmatch}

It is now our task to consider what properties the $\{J_i\}$ must satisfy in order to allow perfect transfer from $\ket{\underline{1}}\rightarrow\ket{\underline{N}}$.
\begin{lemma}
In order to transfer an excitation from $\ket{\underline{1}}\rightarrow\ket{\underline{N}}$ in a 1D $N$-qubit nearest-neighbor $XX$-coupled system with open boundary conditions, the Hamiltonian must be mirror symmetric, i.e.~$J_n^2=J_{N-n}^2$ and $B_n=B_{N+1-n}$ for all $n$. \label{lem:sym}
\end{lemma}
\begin{proof}
Let the eigenvalues and eigenvectors of $H_1$ be $\lambda_n$ and $\ket{\lambda_n}$ respectively. The initial and final target states are given in terms of the eigenvectors by
\begin{eqnarray}
\ket{\underline{1}}&=&\sum_{n=1}^N\alpha_n\ket{\lambda_n}   \\
\ket{\underline{N}}&=&\sum_{n=1}^N\beta_n\ket{\lambda_n}.
\end{eqnarray}
Since the target is
$$
e^{-iH_1t_0}\ket{\underline{1}}=e^{i\phi}\ket{\underline{N}}
$$
for some phase $\phi$ and time $t_0$, it must be true that
$$
e^{-i\lambda_nt_0}\alpha_n=e^{i\phi}\beta_n    \qquad\forall n.
$$
In particular, this reveals that $|\alpha_n|^2=|\beta_n|^2$, and by raising $H_1$ to an integer power, $m$, we can relate
$$
\bra{\underline{1}}H_1^m\ket{\underline{1}}=\sum_{n=1}^N\lambda_n^m|\alpha_n|^2=\bra{\underline{N}}H_1^m\ket{\underline{N}}.
$$
For $m=1$, this gives that $B_1=B_N$. For $m=2$, we find that $B_1^2+J_1^2=B_N^2+J_{N-1}^2$, and thus $J_1^2=J_{N-1}^2$. Each time that $m$ is increased by 1, one new variable is introduced on each side of the equation (either a $B_n$ or $J_n^2$ depending on the parity of $m$), and one finds that they must be equal, giving the required symmetry properties.
\end{proof}
This proof can be extended to more general coupling terms of the form
$$
\text{Re}(J_n)(X_nX_{n+1}+Y_nY_{n+1})+\text{Im}(J_n)(X_nY_{n+1}-Y_nX_{n+1}),
$$
and one finds the condition $|J_n|^2=|J_{N-n}|^2$. The effect of complex coupling coefficients is just to alter the ultimate arrival phase of the $\ket{1}$ component of $\ket{\psi}$ relative to the $\ket{0}$ \cite{Kay:2005b}, so it suffices to restrict to considering real, positive, $J_n$.

Given that $B_n=B_{N+1-n}$ and $J_n=J_{N-n}$, the Hamiltonian $H_1$ commutes with the symmetry operator
$$
S=\sum_{n=1}^N\ket{\underline{n}}\bra{\underline{N+1-n}}.
$$
Thus, $H_1$ further subdivides into symmetric and antisymmetric subspaces, with eigenvectors $\ket{\lambda_n^s}$ and $\ket{\lambda_n^a}$ respectively.

\begin{lemma}
A necessary and sufficient condition for perfect state transfer in a symmetric chain is that there should exist a time $t_0$ and a phase $\phi$ such that $e^{-it_0\lambda_n^s}=e^{i\phi}$ and $e^{-it_0\lambda_n^a}=-e^{i\phi}$ for all eigenvectors $\braket{\lambda_n}{\underline{1}}\neq0$ \cite{shi,Kay:2004c}.
\label{lemma:symmetry_match}
\end{lemma}
\begin{proof}
Consider the decomposition of $\ket{\underline{1}}$ in terms of eigenvectors, and evolve it with the Hamiltonian.
\begin{eqnarray}
e^{-iH_1t_0}\ket{\underline{1}}&=&\sum_ne^{-i\lambda_n^st_0}\alpha_n^s\ket{\lambda_n^s}+e^{-i\lambda_n^at_0}\alpha_n^a\ket{\lambda_n^a}   \nonumber\\
&=&e^{i\phi}\left(\sum_n\alpha_n^s\ket{\lambda_n^s}-\alpha_n^a\ket{\lambda_n^a}\right)=e^{i\phi}S\ket{\underline{1}}. \nonumber
\end{eqnarray}
This immediately reveals the required conditions for any $\alpha_n\neq0$. It follows that all initial states $\ket{\underline{n}}$ transfer to their mirror opposite $S\ket{\underline{n}}=\ket{\underline{N+1-n}}$. Furthermore, if state transfer occurs in time $t_0$, then by symmetry there is a revival of the state on the input spin at times $2nt_0$ for integer $n$ and perfect state transfer at all times $(2n+1)t_0$.
\end{proof}

This represents a remarkably simple condition on the eigenvalues of $H_1$; one which is readily tested, and widely applicable (with slight modification, it need not be limited to a nearest-neighbor Hamiltonian, nor to an excitation preserving Hamiltonian). Within the single excitation subspace of a nearest-neighbor model the $\alpha_n\neq0$ condition can be relaxed by observing that this never happens. For an eigenvector $\ket{\lambda_n}=\sum_m\lambda_{n,m}\ket{\underline{m}}$, the $\lambda_{n,m}$ obey the recursion relation
$$
J_m\lambda_{n,m+1}=(\lambda_n-B_n)\lambda_{n,m}-J_{n-1}\lambda_{n,m-1},
$$
so if $\lambda_{n,m-1}=\lambda_{n,m}=0$, then $\lambda_{n,m+1}=0$. The special case of $m=1$ shows that if $\lambda_{n,1}=\alpha_n=0$, then $\lambda_{n,2}=0$, and so it follows by induction that the entire eigenvector is 0 i.e.~there is no eigenvector with $\alpha_n=0$. Alternatively, we would find that somewhere in the chain, there is a $J_m=0$, but this automatically forbids state transfer. The only remaining problem is how to identify which eigenvalues belong to the set $\{\lambda_n^s\}$, and which to $\{\lambda_n^a\}$.

\begin{lemma}
For tridiagonal matrices, with negative off-diagonal matrix elements and eigenvalues $\lambda_n$ ($n=1\ldots N$), the number of sign changes in the eigenvectors $\ket{\lambda_n}$ is $n-1$. The eigenvectors are assumed to be ordered such that $\lambda_n<\lambda_{n+1}$.
\end{lemma}
This is proven in \cite{gladwell}. As a consequence, for symmetric tridiagonal matrices, the symmetry of the eigenvectors alternates, as it is determined by the number of sign changes (an odd number means that the eigenvector is antisymmetric). Our matrix $H_1$ is a tridiagonal matrix, but with positive coefficients $J$. As such, the ordering is reversed, with the maximum eigenvector being the one with no sign changes. Consider $\{\lambda_n\}$ to be an ordered set of eigenvalues, $\lambda_n<\lambda_{n+1}$. The necessary and sufficient conditions for state transfer in a symmetric chain then become that
\begin{equation}
\lambda_n-\lambda_{n-1}=(2m_n+1)\pi/t_0, \label{eqn:st_cond}
\end{equation}
where $t_0$ is the state transfer time, and $m_n$ is a positive integer (which can vary with $n$). This means that the ratio of pairs of differences of eigenvalues must be rational.




\subsection{Inverse Eigenvalue Problems} \label{sec:iep}

In Sec.~\ref{sec:symmatch}, we saw that in order to test if perfect state transfer can be achieved in a 1D nearest-neighbor coupled chain, one only needs to test a symmetry condition, and then consider the eigenvalues of an $N\times N$ matrix. Rather than simply test such a condition, we would like to know how to fix the coupling strengths $\{J_n\}$ and $\{B_n\}$. Fortunately, this is an extremely well studied problem. We can select any set of (non-degenerate) eigenvalues $\{\lambda_n\}$ compatible with Eqn.~(\ref{eqn:st_cond}), and solve an Inverse Eigenvalue Problem \cite{gladwell} to find, in time $\text{poly}(N)$, the corresponding coupling strengths. Every possible choice of eigenvalues has a unique solution with positive (non-zero), symmetric, $J_n$.

Specifying a chain by its eigenvalues is mathematically elegant, but why should one select a particular set of eigenvalues? In the following subsections, we shall consider a specific set of eigenvalues, and show that they yield a family of chains which are optimal for a range of properties, such as transfer speed. For the remainder of this subsection, however, we shall outline a technique of Karbach and Stolze \cite{transfer_comment} that motivates how to select the eigenvalues so that the final chain has specific properties.

The specific target that we have in mind is to create state transfer chains with as little spatial variation of the couplings as possible. One might hope that this would reduce manufacturing requirements for such chains\footnote{In practice, we will end up trading the degree of spatial variation for the required accuracy of the eigenvalues, and thus coupling strengths.}. As such, we start from a situation where all the coupling strengths are equal, $J_n=1$, and $B_n=0$, which would be our ideal situation, except that perfect transfer is impossible. The eigenvalues are readily found:
$$
\lambda_n=-2\cos\left(\frac{n\pi}{N+1}\right).
$$
Let us denote $\min_n(\lambda_n-\lambda_{n-1})$, the minimum spacing in the spectrum, by $\delta$. We now select a quantity $\pi/t_0\ll\delta$, which will determine the state transfer time. In the present case, this means that $t_0\sim N^2$. Each of the eigenvalues can then be changed by no more than $\pi/t_0$ such that the perfect state transfer condition of Eqn.~(\ref{eqn:st_cond}) is satisfied. Subsequently, an Inverse Eigenvalue Problem can be solved to find the coupling strengths of a chain with that spectrum. Since the eigenvalues were only slightly perturbed, the coupling strengths are only slightly perturbed, and thus the resultant chain is close to being uniformly coupled, and close to having no magnetic field. One can improve on this by imposing on the choice of spectrum that $\lambda_n=-\lambda_{N+1-n}$, so the magnetic fields necessarily come out to zero.

\begin{lemma}
By selecting the eigenvalues of a symmetric tridiagonal matrix to satisfy $\lambda_n=-\lambda_{N+1-n}$ for all $n$, the result of the inverse eigenvalue problem automatically satisfies $B_m=0$ for all $m$. \label{lemma:zero}
\end{lemma}
\begin{proof}
We will prove this by giving the essence of the inverse eigenvalue algorithm \cite{hochstadt}\footnote{While containing the essence of the solution to the inverse eigenvalue, this particular method is not the most numerically stable.}. Consider the function
$$
B'(\lambda_n)=\prod_{m=1\neq n}^N(\lambda_n-\lambda_m),
$$
which is the derivative of the function $B(\lambda)=\prod_{m=1}^N(\lambda-\lambda_m)$, the characteristic polynomial of the system, evaluated at $\lambda_n$.
This can be expanded in the case of $\lambda_n=-\lambda_{N+1-n}$ such that
\begin{eqnarray}
B'(\lambda_n)&=&-2\lambda_n\left(\prod_{m=1\neq n}^{N/2}(\lambda_m-\lambda_n)(-\lambda_m-\lambda_n)\right)  \nonumber\\
B'(\lambda_{N+1-n})&=&2\lambda_n\left(\prod_{m=1\neq n}^{N/2}(\lambda_m+\lambda_n)(-\lambda_m+\lambda_n)\right)  \nonumber
\end{eqnarray}
for even $N$. Odd $N$ follows similarly, so that, for any $N$, $B'(\lambda_n)=(-1)^{N-1}B'(\lambda_{N+1-n})$.
In \cite{hochstadt}, it is proven that for a symmetric tridiagonal matrix,
\begin{equation}
|\alpha_n|^2=\frac{1}{(-1)^nB'(\lambda_n)\sum_{m=1}^N\frac{1}{(-1)^nB'(\lambda_m)}}.  \label{eqn:hoch}
\end{equation}
Taking the ratio of these terms,
$$
\frac{|\alpha_n|^2}{|\alpha_{N+1-n}|^2}=\frac{(-1)^{N-1}B'(\lambda_{N+1-n})}{B'(\lambda_n)}=1.
$$
This imposes that
$$
\bra{\underline{1}}H_1^{2m+1}\ket{\underline{1}}=\sum_{n=1}^N\lambda_n^{2m+1}|\alpha_n|^2=0
$$
for every integer $m$. We must now compare this to an expansion of $\bra{\underline{1}}H_1^{2m+1}\ket{\underline{1}}$ in terms of $J_n$ and $B_n$. Every term in this sum must contain an odd number of $B_n$ with $n\leq m$, and only one term $J_1^2\ldots J_{m-1}^2B_m$ contains $B_m$. Provided $J_n\neq 0$, we can therefore solve iteratively to find that all $B_n=0$.
\end{proof}

\subsection{An Analytic Solution} \label{sec:anal}

Rather than relying on numerical solutions, it is often beneficial to have special cases enumerated. A number of coupling schemes with integer spectra have been found \cite{Christandl,lambrop,Christandl:2004a,shi,xi}, although we will only concentrate on the simplest of these, which is based on the $J_x$ rotation matrix of a spin $\half(N-1)$ particle. This can be written as a tridiagonal matrix
\begin{equation}
J_x=\half\sum_{n=1}^N\sqrt{n(N-n)}(\ket{\underline{n}}\bra{\underline{n+1}}+\ket{\underline{n+1}}\bra{\underline{n}}),  \label{eqn:JX}
\end{equation}
which is exactly the form of $H_1$ ($J_n=\half\sqrt{n(N-n)}$, $B_n=0$). You may recall that the eigenvectors must be the $\ket{M_x}$ states, with eigenvalues $M_x=\half(N-1),\half(N-1)-1\ldots-\half(N-1)$. Thus, the eigenvalues all satisfy $\lambda_n-\lambda_{n-1}=1$, and since $[J_x,S]=0$, the conditions for perfect transfer are satisfied, with $t_0=\pi$. To see this in a more natural way, consider the evolution of a Hamiltonian
$$
H_X=\half\sum_{n=1}^{N-1}X_n
$$
acting on the initial state $\ket{0}^{\otimes N-1}$. Clearly, in a time $\pi$, the state evolves to $\ket{1}^{\otimes N}$. Now let us introduce the states
$$
\ket{\psi_n}=\frac{1}{\sqrt{\binom{N-1}{n}}}\sum_{\stackrel{x\in\{0,1\}^{N-1}}{w_x=n}}\ket{x},
$$
where $w_x$ is the Hamming weight of the bit string $x$. One simply has to verify that
$$
H\ket{\psi_n}=J_{n-1}\ket{\psi_{n-1}}+J_n\ket{\psi_{n+1}}
$$
to see that it takes on the claimed diagonal form of $J_x$ on the subspace of $\{\ket{\psi_n}\}$, and that the claimed perfect transfer is achieved. Alternative proofs can be found in \cite{shore, Christandl,Kay:2004c,Christandl:2004a,lambrop}.

\subsection{Optimality of Solutions} \label{sec:optimal}

Evidently, arbitrarily fast transfer can be achieved by just continuing to scale up the coupling strengths. A transformation $J_n\rightarrow \kappa J_n$ for all $n$ yields $t_0\rightarrow t_0/\kappa$. However, in practical settings, we are typically constrained by properties such as the maximum coupling strength, or possibly the maximum eigenvalue. In such situations, which is the best set of eigenvalues to choose? We will now show that the optimal choice with respect to a number of conditions is that of Eqn.~(\ref{eqn:JX}), as proved in \cite{Kay:2005e,yung:06}.

\subsubsection{Transfer Time}

We would like to optimize the transfer time $t_0$ subject to some constraint. The most easily applied constraint is that of the maximum eigenvalue. Alternatively, this can be viewed as, for a fixed $t_0$, how can we minimize the maximum eigenvalue? The minimum spacing between eigenvalues is $\pi/t_0$, and this is achieved by the analytic solution previously given. Thus, the difference $\lambda_{\max}-\lambda_{\min}$ is a minimum. While bounding this range is useful, it is not the most physical parameter to optimize with respect to. Rather, it would be preferable to find, for a given maximum coupling strength $J_{\max}$, what the minimum transfer time is. This problem was first considered in \cite{yung:06}. We will concentrate on the case of even $N$, as the odd $N$ case is somewhat more involved.

We define $\delta_k=\lambda_{k+1}-\lambda_{k}$ to be the difference in eigenvalues, and $\delta=\min_k\delta_k$. Note that $\delta\geq \pi/t_0$, and we wish to consider all sets of models with the same $t_0$. The Hamiltonian divides into symmetric and antisymmetric subspaces, with the eigenvalues satisfying
$$
\sum_n(\lambda_n^s-\lambda_n^a)=\Tr(SH_1)=2J_{N/2}.
$$
For any particular coupling scheme, one has
$$
J_{\max}\geq J_{N/2}\geq\smallfrac{4} N\delta\geq\smallfrac{4} N\pi/t_0.
$$
Moreover, for the analytic scheme, these inequalities completely collapse, $J_{\max}=\smallfrac{4} N/t_0$, and thus, for a fixed $t_0$, this represents the scheme with the smallest $J_{\text{max}}$.

\subsubsection{Timing Errors}

\begin{figure}[!tb]
\begin{center}
\includegraphics[width=0.45\textwidth]{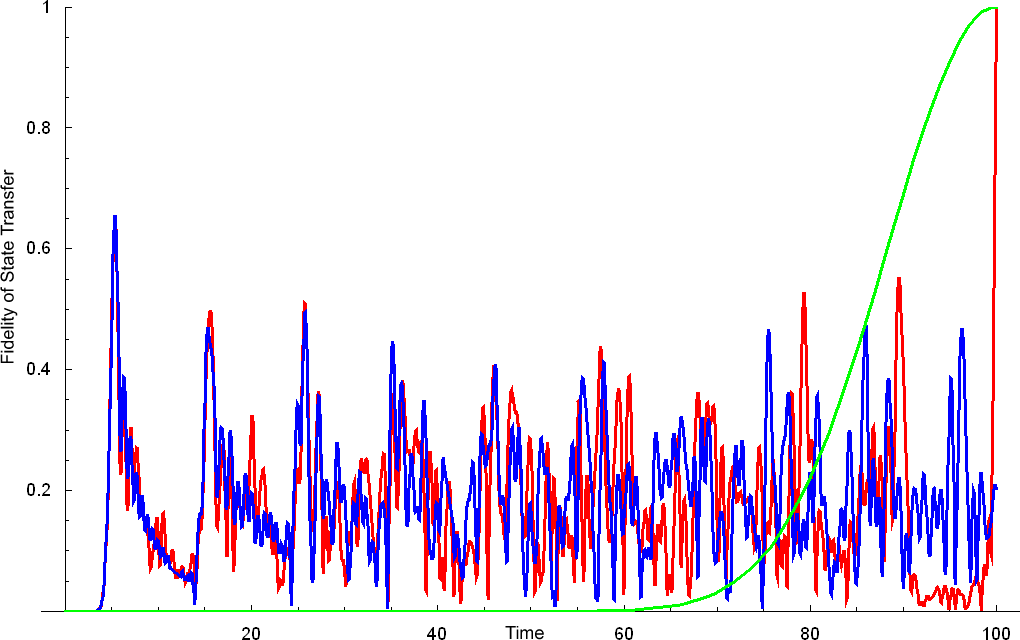}
\caption{A comparison of the transfer fidelity $\bra{\psi_{\text{in}}}\rho_{\text{out}}\ket{\psi_{\text{in}}}$ of the uniformly coupled chain of 31 qubits (blue), the almost uniformly coupled chain giving perfect state transfer (red), and the analytic solution (green).} \label{fig:errors}
\end{center}
\end{figure}

One of the underlying assumptions to the process of state transfer has been that at $t=0$, a quantum state can be placed on spin 1, and at $t=t_0$, it can be removed from spin $N$, instantaneously. In practice, this will not be the case, so it would be useful to know how to design the chain so that for a given tolerable error $\varepsilon$, the possible time window for implementing the transfer of the state is as long as possible. Critical to this consideration are the quantities
\begin{equation}
\gamma_n(t)=\bra{\underline{n}}e^{-iH_1t}\ket{\underline{1}},   \label{eqn:gamma}
\end{equation}
and, more specifically, $\gamma_1(t)$. Given the symmetry of the chain, optimizing this function will also optimize the equivalent function for the arrival of the state at the other end of the chain. Let us expand $\gamma_1(t)$ in terms of the eigenvectors $\{\ket{\lambda_n}\}$ of $H_1$, $\gamma_1(t)=\sum_{n=1}^N|\alpha_n|^2e^{-i\lambda_nt}$. Of course, $\gamma(0)=1$, and this must be a maximum. Thus, how quickly it changes is governed by
$$
\left.\frac{d^2\gamma_1}{dt^2}\right|_{t=0}=-\sum_{n=1}^N|\alpha_n|^2\lambda_n^2=-\bra{\underline{1}}H_1^2\ket{\underline{1}},
$$
and minimizing the rate of departure of the excitation equates to minimizing $J_1^2+B_1^2$ for a given $t_0$. It thus seems likely that our analytic solution, where $B_1=0$ and $J_1$ is only $1/\sqrt{N}$ of the maximum coupling strength, will be extremely robust. For example, Fig.~\ref{fig:errors} compares this chain with another following the technique of Sec.~\ref{sec:iep}. Further confirmation arises from a closer examination of the expression $\sum_{n=1}^N|\alpha_n|^2\lambda_n^2$. The analytic solution has the smallest possible maximum eigenvalue. Furthermore, upon solving for the eigenvectors \cite{Christandl:2004a}, one finds that the $\alpha_n$ corresponding to the largest $|\lambda_n|$ are exponentially suppressed in $N$, and only those with essentially no contribution to $\lambda_n^2$ have a contribution in $|\alpha_n|^2$, which thus shows that the analytic solution must be very close to optimal, given that none of the $\alpha_n=0$.

\subsection{Beyond Nearest-Neighbour Couplings}

There is no special reason to restrict to nearest-neighbor couplings in the Hamiltonian $H$; in practice there are likely to be residual couplings with strengths that decrease with distance. The introduction of $XX+YY$ terms between non-nearest-neighbors does not change the fact that the Hamiltonian decomposes into distinct excitation subspaces, and the symmetry matching condition of Lemma \ref{lemma:symmetry_match} still holds. The Inverse Eigenvalue Problem for this wider class of cases is less well studied, and the existence of solutions is no longer guaranteed -- the alternating symmetry property of eigenvectors is not guaranteed to hold, which leads to problems with selecting the correct eigenvalues. Nevertheless, by using the technique of Sec.~\ref{sec:iep} to start from a solution where eigenvalues can be calculated, a near-by solution can be expected to exist, and then an efficient algorithm has been formulated to find the coupling parameters \cite{Kay:2005e}, as will now be described.

We are given a Hamiltonian $H({\vec{r}})$ which satisfies
$$
\left[H({\vec{r}}),\sum_{n=1}^NZ_n\right]=0,
$$
and is represented by a symmetric $N\times N$ matrix $H_1({\vec{r}})$ in the first excitation subspace. This Hamiltonian depends on $N$ parameters $\{r_i\}$, which can be thought of, for example, as the positions of the spins when they are coupled by a distance-dependent coupling strength. The desired eigenvalues are contained in the $N\times N$ diagonal matrix $\Lambda$\footnote{Note that one must select a consistent scheme for ordering the eigenvalues on the diagonal.}.

The algorithm starts with a first estimate ${\overrightarrow{r_{(0)}}}$ for ${\vec{r}}$. The matrix $H_1({\overrightarrow{r_{(0)}}})$ is diagonalized by $U_0$,
$$
H_1({\overrightarrow{r_{(0)}}})=U_0\Lambda(\identity+\epsilon \Delta_0)U_0^\dagger,
$$
where $\Delta_0$ is a diagonal matrix which encapsulates the errors in the energies, and $\epsilon$ is a small parameter. The second step of the algorithm should use a vector ${\overrightarrow{r_{(1)}}}={\overrightarrow{r_{(0)}}}+\epsilon\ {\overrightarrow{\delta r}}$, and it is our task to find the best choice for ${\overrightarrow{\delta r}}$. A unitary $U_1$ diagonalizes the Hamiltonian,
\begin{equation}
H_1({\overrightarrow{r_{(1)}}})=U_1\Lambda(\identity+\epsilon \Delta_1)U_1^\dagger.
\label{eqn:iep}
\end{equation}
Provided $\epsilon$ is small, $U_0$ and $U_1$ should be close to each other. As such, $U_1$ can be parametrized in terms of $\epsilon$,
\begin{equation}
U_1=U_0(\identity+i\epsilon Q)(\identity-i\epsilon Q)^{-1},	\label{eqn:param}
\end{equation}
where $Q$ is a Hermitian matrix containing information on the change in eigenvectors. This parametrization ensures that $U_0^\dagger U_1\rightarrow \identity$ as $\epsilon\rightarrow 0$, and the unitarity of $U_1$ can be seen by rewriting
$$
i\epsilon Q=(\identity+U_1^\dagger U_0)^{-1}(\identity-U_1^\dagger U_0).
$$
Given that the eigenvalues of $U_1^\dagger U_0$ must be of the form $e^{-i\theta}$, the eigenvalues of $Q$ are real ($\tan(\theta/2)$) i.e.~$Q$ is Hermitian. Conversely, for any Hermitian operator $Q$, the eigenvalues can be identified with $\tan(\theta/2)$, and thus $U_1^\dagger U_0$ must be unitary. The parametrization of Eqn.~(\ref{eqn:param}) can now be substituted into Eqn.~(\ref{eqn:iep}), and expanded in terms of $\epsilon$. The terms for $\epsilon^0$ cancel, and the terms for $\epsilon^1$ give
$$
\sum_n\delta r_nU_0^\dagger\left.\frac{\partial H_1}{\partial r_n}\right|_{\overrightarrow{r_{(0)}}}U_0=\Lambda \Delta_1-\Lambda \Delta_0+2i(Q\Lambda-\Lambda Q).
$$
The aim of the iteration should be to choose ${\overrightarrow{\delta r}}$ such that $\Delta_1=0$. Since the diagonal elements of the final term, $Q\Lambda-\Lambda Q$, are zero, all the eigenvalue information is encapsulated by the diagonal elements of the equations, while changing eigenvectors only affects the off-diagonal elements. Thus, the previous equation can be rewritten for just the diagonal elements,
\begin{equation}
M.{\overrightarrow{\delta r}}={\vec{b}},
\label{eqn:inverse}
\end{equation}
where ${\vec{b}}$ is a vector of the diagonal elements of $-\Lambda \Delta_0$, and the $n^{th}$ column of the matrix $M$ is given by the diagonal elements of
$$
U_0^\dagger\left.\frac{\partial H_1}{\partial r_n}\right|_{\overrightarrow{r_{(0)}}}U_0.
$$
The solution to Eqn.~(\ref{eqn:inverse}) is the vector ${\overrightarrow{\delta r}}$, which gives the correct eigenvalues to $O(\epsilon^2)$. Provided $\overrightarrow{\delta r}$ is small in comparison to $\overrightarrow{r_{(0)}}$, this protocol iterates, squaring the error at each step. Hence, to achieve an accuracy of $\epsilon_0$, only $O(\log(\epsilon_0))$ iterations are needed.
 Since there are efficient algorithms for solving Eqn.~(\ref{eqn:inverse}), and because the matrices to be diagonalized are symmetric (hence there are efficient diagonalization procedures, such as Householder reductions \cite{numerical_recipes}), the cost of each iteration scales polynomially, $O(N^3)$, with the number of qubits in the chain \cite{numerical_recipes}, and the required parameters can be found with an efficient classical computation.

\subsection{Alternative State Transfer Systems} \label{sec:harm}

In this section, we have derived the necessary and sufficient conditions for perfect state transfer on a chain of spins coupled by an $XX$ Hamiltonian, and found an analytic solution which is optimal with respect to properties such as state transfer time. These same conditions can be applied to a number of other systems which are excitation preserving, although it is important to note that in the next section, when we move to using many excitations, the equivalences no longer hold.

The first case that we shall consider is a modulated anisotropic Heisenberg Hamiltonian
\begin{eqnarray}
H_{\text{heis}}&=&\half\sum_{n=1}^{N-1}J_n(X_nX_{n+1}+Y_nY_{n+1}+\Delta_nZ_nZ_{n+1})	\nonumber\\
&&-\sum_{n=1}^Nb_nZ_n.	\nonumber
\end{eqnarray}
Within the first excitation subspace, the original matrix $H_1$ can be recovered by setting
$$
b_n=B_n+\half(J_n\Delta_n+J_{n-1}\Delta_{n-1}),
$$
up to an identity matrix. One consequence of this, however, is that it is impossible to achieve perfect state transfer with all the $b_n$ set to 0, unlike lemma \ref{lemma:zero} \cite{marcin}.

The other case of interest is a chain of harmonic oscillators \cite{plenio:0}, although this can be further generalized to a chain of $q$-deformed oscillators \cite{innocent}. As before, there is a set of $N$ sites, and at each site $n$ we introduce the creation and annihilation operators $a^\dagger_n$ and $a_n$, which obey the bosonic commutation properties $[a_n,a_m^\dagger]=\identity\delta_{n,m}$ and $[a_n,a_m]=[a_n^\dagger,a_m^\dagger]=0$. A nearest-neighbor Hamiltonian of the form
\begin{equation}
H_{\text{ho}}=2\sum_{n=1}^NB_na^\dagger_na_n+\sum_{n=1}^{N-1}J_n(a_{n+1}^\dagger a_n+a^\dagger_na_{n+1})	\label{eqn:ho}
\end{equation}
is also excitation preserving,
$$
\left[H_{\text{ho}},\sum_{n=1}^Na_n^\dagger a_n\right]=0.
$$
Starting from the vacuum state, $\ket{\underline{0}}$, the single excitation subspace can be defined by $\ket{\underline{n}}=a_n^\dagger\ket{\underline{0}}$, and the Hamiltonian within this subspace is the same form as $H_1$. Thus, all properties of perfect state transfer are recovered.

\section{Higher Excitation Subspaces} \label{sec:highexcit}

\subsection{The Jordan-Wigner Transformation} \label{sec:JW}

The primary benefit of choosing to examine the nearest-neighbor $XX$ model rather than the Heisenberg model is that when we come to analyze the presence of many excitations on the chain, vast simplifications arise. This is because one can apply the Jordan-Wigner transformation \cite{JordanWigner}, which maps the excitations on the spins into {\em non-interacting} fermions. This mean that the individual fermions behave exactly as they did in the single excitation subspace, and we only have to take care of the exchange phase.

The Jordan-Wigner transformation is as follows. We introduce creation operators $a_n^\dagger$ for fermions at site $n$, and identify them with
$$
a_n^\dagger=\half\left(\prod_{m=1}^{n-1}Z_m\right)(X-iY)_n.
$$
Thus, one can rewrite the Hamiltonian $H$ as
$$
H_{JW}=\sum_{n=1}^{N-1}J_n(a^\dagger_n a_{n+1}+a^\dagger_{n+1}a_n)+2\sum_{n=1}^NB_na_n^\dagger a_n,
$$
up to an irrelevant $\identity$ term. Let's now take the ground state of the system, $\ket{\underline{0}}$, the state with no fermions. The eigenvectors of $H_{JW}$ in the one fermion sector can be written as
$$
\sum_{m=1}^N\lambda_{n,m}a_m^\dagger\ket{\underline{0}}
$$
with eigenvalues $\lambda_n$. Our claim that the $H_{JW}$ corresponds to non-interacting fermions needs to be shown by diagonalizing the full Hamiltonian, not just the one fermion subspace. We therefore introduce new operators $b_n^\dagger$, and propose that they satisfy
$$
b_n^\dagger=\sum_{m=1}^N\lambda_{n,m}a_m^\dagger,
$$
with the intent to show that 
$$H_{JW}=\sum_n\lambda_nb_n^\dagger b_n.$$
Given that the $\{\lambda_{n,m}\}$ define the eigenvectors in the first excitation subspace, they must satisfy
$$
\lambda_n\lambda_{n,m}=B_m\lambda_{n,m}+J_{m-1}\lambda_{n,m-1}+J_m\lambda_{n,m+1}.
$$
This helps when we expand the $b_n$ in $\sum_n\lambda_nb_n^\dagger b_n$,
$$
\sum_{n,m,k=1}^N(B_m\lambda_{n,m}+J_{m-1}\lambda_{n,m-1}+J_m\lambda_{n,m+1})a_m^\dagger \lambda_{n,k}^*a_k.
$$
Now we use an orthogonality relation for the eigenvectors,
$$
\bra{\underline{0}}a_ma_k^\dagger\ket{\underline{0}}=\sum_{n=1}^N\lambda_{n,m}\lambda_{n,k}^*=\delta_{m,k},
$$
which reveals that
\begin{eqnarray}
\sum_n\lambda_nb_n^\dagger b_n&=&\sum_{m=1}^NB_ma_m^\dagger a_m+J_{m-1}a_m^\dagger a_{m-1}+J_ma_m^\dagger a_{m+1}	\nonumber\\
&=&H_{JW}.	\nonumber
\end{eqnarray}
Therefore, the $\{b_n^\dagger\}$ describe independent fermions. Heuristically, this means that if excitations on either spins 1 or 2 (say) transfer perfectly, then a pair of excitations on both spins 1 and 2 transfer perfectly to a pair of excitations on spins $N-1$ and $N$ because they are independent. However, because the two fermions have exchanged (the one on spin 1 has gone past the one on spin 2), a multiplicative factor of $-1$ is also introduced.

The wedge product \cite{Osborne,Kay:2005e} is perhaps the most useful tool that we can bring to bear as a result of the Jordan-Wigner transformation. This is used to combine states of single excitations into valid states of multiple excitations. For example, $\ket{\underline{1}}\wedge\ket{\underline{2}}$ denotes excitations on spins 1 and 2. It has two crucial properties that convey the fermionic nature of the excitations. The first is exchange, $\ket{\underline{1}}\wedge\ket{\underline{2}}=-\ket{\underline{2}}\wedge\ket{\underline{1}}$, and the second is the exclusion principle, $\ket{\underline{1}}\wedge\ket{\underline{1}}=0$. As a result of the exchange property, we must pick a standard ordering for terms; the excitations are written in order from left to right, so a state of 5 qubits $\ket{01001}$ along a chain is written as $\ket{\underline{2}}\wedge\ket{\underline{5}}$.

From these basic properties of the wedge product, further properties are readily derived. Consider two vectors $\ket{a}=\sum_na_n\ket{\underline{n}}$ and $\ket{b}=\sum_nb_n\ket{\underline{n}}$, which are correctly normalized.
$$
\ket{a}\wedge\ket{b}=\sum_{m<n}(a_mb_n-a_nb_m)\ket{\underline{m}}\wedge\ket{\underline{n}}
$$
The normalization of the new state is
\begin{eqnarray}
\half\sum_{m,n}|a_mb_n-a_nb_m|^2&=&\sum_{m,n}\left(|a_m|^2|b_n|^2-a_mb_m^*a_n^*b_n\right)   \nonumber\\
&=&1-|\braket{b}{a}|^2. \nonumber
\end{eqnarray}
Again, we find that the two excitations cannot be in the same state, $\ket{a}\wedge\ket{a}=0$, although this is no longer restricted to the computational basis. The eigenvectors in, for example, the second excitation subspace, can be calculated from $\ket{\lambda_n}\wedge\ket{\lambda_m}$, which is equivalent to evaluating the Slater determinant \cite{Christandl:2004a}. The corresponding eigenvalues are $\lambda_n+\lambda_m$. As such, the evolution due to the Hamiltonian $H$ can be found from
$$
\left(e^{-iH_1t}\ket{a}\right)\wedge\left(e^{-iH_1t}\ket{b}\right).
$$
To verify that this statement is consistent, we expand $\ket{a}=\sum_na_n\ket{\lambda_n}$ and $\ket{b}=\sum_nb_n\ket{\lambda_n}$;
\begin{eqnarray}
e^{-iHt}(\ket{a}\wedge\ket{b})&=&e^{-iHt}\sum_{n,m}a_nb_m\ket{\lambda_n}\wedge\ket{\lambda_m}  \nonumber\\
&=&\sum_{n,m}a_nb_me^{-it(\lambda_n+\lambda_m)}\ket{\lambda_n}\wedge\ket{\lambda_m}    \nonumber\\
&=&\sum_{n,m}\left(a_ne^{-i\lambda_nt}\ket{\lambda_n}\right)\wedge\left(b_me^{-i\lambda_mt}\ket{\lambda_m}\right)    \nonumber\\
&=&\left(e^{-iH_1t}\ket{a}\right)\wedge\left(e^{-iH_1t}\ket{b}\right).    \nonumber
\end{eqnarray}

\subsection{Entanglement Generation} \label{sec:entgen}

One of the most trivial applications of multiple excitations on a spin chain is in the dynamic generation of entanglement. Consider starting a perfect state transfer chain in the state $\ket{+}\ket{0}^{\otimes N-2}\ket{+}$, where $\ket{+}=(\ket{0}+\ket{1})/\sqrt{2}$, and allow it to evolve for a time $t_0$. In that time, the two $\ket{+}$ states exchange positions. However, only when two excitations are present do they generate an exchange phase. Thus, the system evolves from
$$
\half\left(\ket{\underline{0}}+\ket{\underline{1}}+\ket{\underline{N}}+\ket{\underline{1}}\wedge\ket{\underline{N}}\right)
$$
into the state
$$
\half\left(\ket{\underline{0}}+\ket{\underline{N}}+\ket{\underline{1}}+\ket{\underline{N}}\wedge\ket{\underline{1}}\right).
$$
Restoring the normal ordering of the operators, and neglecting the central $N-2$ spins, the final state is
$$
\half(\ket{00}+\ket{01}+\ket{10}-\ket{11}),
$$
which is a maximally entangled state between the two most distant points. This simple observation led to the proposal that perfect state transfer chains be used to generate 1D cluster states in a highly parallel process \cite{Jaksch:2004a}.

\subsection{System Initialization} \label{sec:sysinit}

The work on multiple excitations can also be used to remove the requirement of initializing the spin chain in the $\ket{\underline{0}}$ state before the transfer. Let us encode the unknown state $\ket{\psi}$ into the first two spins of the chain, as the state $\alpha\ket{01}+\beta\ket{10}$ \cite{kay-2006b}. We shall assume that the rest of the chain is in some unknown state $\ket{\vec x}$, where ${\vec x}$ is a string of $N-2$ bits denoting the position of excitations on the chain. If the transfer scheme works for all possible states $\ket{\vec x}$, then it also works for all possible superpositions and mixtures, $\rho$. We are only interested in transferring the state $\ket{\psi}$, and thus, there is no requirement to preserve the initial state of the chain.

At the start of the state transfer process, we have the state
$$
(\alpha\ket{\underline{1}}+\beta\ket{\underline{2}})\wedge\ket{\vec x}_{3\ldots N},
$$
and this transfers in time $t_0$ to the state
$$
(\alpha\ket{\underline{N}}+\beta\ket{\underline{N-1}})\wedge\ket{\vec x^T}_{1\ldots N-2},
$$
where $\ket{\vec x^T}$ is the mirror of state $\ket{\vec x}$, incorporating any phase of $(-1)$ that may result. Reasserting the normal ordering of the wedge product, this gives
$$
(-1)^{w_x}\ket{\vec x^T}_{1\ldots N-2}\wedge(\alpha\ket{\underline{N}}+\beta\ket{\underline{N-1}}),
$$
and therefore the encoded state $\ket{\psi}$ is transferred perfectly. For general initial states $\rho$ on the rest of the chain, the output state is not the re-ordered form of $\rho$ because of the different phases appearing on the vectors $\ket{\vec x^T}$, but this is of no concern to us. Heuristically, the process that we have gone through is to note that the only problem when transferring states in higher excitation subspaces is that exchange phases are generated, and may cause entanglement. However, this entanglement is generated between different excitation subspaces (yielding different numbers of exchange phases), and thus, by encoding within a subspace of fixed parity of excitation number, the state to be transferred does not become entangled.

Although encoding this state requires a (nearest-neighbor) entangling operation, the eventual read-out does not; one can measure spin $N-1$ in the $X$ basis, and remove the state from spin $N$ (up to a $Z$ rotation, depending on the measurement outcome).

In Sec.~\ref{sec:amp}, we will describe an alternative transfer scheme which does not require any encoding of the quantum state to be transmitted, and is nevertheless insensitive to the initial state of the chain. This is achieved by a unitary rotation applied to existing transfer schemes such that the unitary maps the encoded states described in this section to un-encoded states in the new scheme. The transverse Ising model \cite{paternostro:1,paternostro:2} also achieves this feat of not requiring an encoding by virtue of a special feature of the way it can be mapped into the $XX$ model that we considered here, as we will see in Sec.~\ref{sec:bogoliubov} (thereby explaining the close link between the coupling strengths and magnetic fields found in \cite{paternostro:1,paternostro:2} and the analytic solution we gave in Sec.~\ref{sec:anal}). On the other hand, the bosonic system of Eqn.~\ref{eqn:ho} does not have an entangling operation between different particles and hence works directly in higher excitation subspaces, without the need for an encoding.

\subsection{Transfer Rate} \label{sec:transrate}

The state transfer protocol involves placing an unknown state $\ket{\psi}$ on spin 1, waiting a time $t_0$, and finding the state perfectly transferred to spin $N$. With the added ability to consider higher excitation subspaces, we can now consider the state transfer chain as a channel for communicating quantum states \cite{Kay:2005e}, and then the interesting question is what rate of transmission of single-qubit states can be achieved?

Let us start with a state $\ket{\psi_1}=\alpha_1\ket{0}+\beta_1\ket{1}$ on spin 1, and let it evolve for a time $t_r$. At this time, we will introduce a state $\ket{\psi_2}$ onto spin 1 (tracing out the state of that spin). After further evolution for $t_0-t_r$, we might hope that a high quality version of $\ket{\psi_1}$ would appear at spin $N$, essentially unaffected by the introduction of $\ket{\psi_2}$. If so, this would allow transmission at a rate $1/t_r$, yielding a significant improvement over the original $1/t_0$ rate.

At time $t_r$, before $\ket{\psi_2}$ is added, the system is in the state $\alpha_1\ket{\underline{0}}+\beta_1e^{-iH_1t_r}\ket{\underline{1}}$. The state can be split into parts which are $\ket{0}$ and $\ket{1}$ on spin 1,
$$
e^{-iH_1t_r}\ket{\underline{1}}=(e^{-iH_1t_r}\ket{\underline{1}}-\gamma_1(t_r)\ket{\underline{1}})+\gamma_1(t_r)\ket{\underline{1}},
$$
which allows spin 1 to be traced out. The trace gives a mixture of two components,
$$
\alpha_1\ket{\underline{0}}+\beta_1(e^{-iH_1t_r}\ket{\underline{1}}-\gamma_1(t_r)\ket{\underline{1}})
$$
and
$$
\beta_1\gamma_1(t_r)\ket{\underline{1}}.
$$
This immediately reveals that the protocol can only be perfect if $\gamma_1(t_r)=0$.
The state $\ket{\psi_2}$ is subsequently added to spin 1, and the system is evolved for time $t_0-t_r$. In that time, the two contributions become
\begin{eqnarray}
&\alpha_1\alpha_2\ket{\underline{0}}+\alpha_2\beta_1\ket{\underline{N}}+(\alpha_1\beta_2-\alpha_2\beta_1\gamma_1(t_r))e^{-iH_1(t_0-t_r)}\ket{\underline{1}}&   \nonumber\\
&+\beta_1\beta_2(e^{-iH_1(t_0-t_r)}\ket{\underline{1}})\wedge\ket{\underline{N}}& \nonumber
\end{eqnarray}
and
$$
(\alpha_2\ket{\underline{0}}+\beta_2e^{-iH_1(t_0-t_r)}\ket{\underline{1}})\gamma_1(t_r)\beta_1.
$$
From this, one can calculate the reduced density matrix $\rho$ of qubit $N$, and then evaluate the fidelity
$$
F=\bra{\psi_1}\rho\ket{\psi_1},
$$
The full expression is too complex to reproduce here, but one can more readily understand an example case, of $\beta_1=\beta_2=1$, where it turns out that
$$
F=1-|\gamma_1(t_r)|^2-|\gamma_1(t_r)|^4.
$$
For the analytic solution of the state transfer chain, $\gamma_1(t)$ is a decreasing function of time ($0\leq t\leq t_0$), with
$$
\gamma_1(t)=\cos^{N-1}\left(\frac{\pi t}{2t_0}\right).
$$
As a result, it is possible to achieve a transfer rate of $\sqrt{N}/t_0$ for some constant error probability $\varepsilon$ on each $\ket{\psi_n}$ that is transferred.

Since our focus is on the perfect realization of these protocols, we will focus on the case $\gamma_1(t_r)=0$. Moreover, to achieve the target rate of $1/t_r$, we will require that $\gamma_1(nt_r)=0$ for all integers $n=1\ldots 2t_0/t_r-1$. It is not always possible to enhance this rate; the smallest non-trivial case is $N=3$ but, as a perfect state transfer system, this is uniquely specified (up to an overall numerical factor of $1/t_0$) if we assume no magnetic fields can be used. One can readily show that $\gamma_1(t<t_0)>0$.

Given that $\gamma_1(t_r)=0$ implies that an initial state $\ket{1}$ is orthogonal to the state at time $t_r$, the Margolus-Levitin Theorem \cite{marg} can be used to bound the achievable rate. This theorem states that the minimum time for one state to evolve into an orthogonal one must be at least
\begin{equation}
t_r\geq\frac{\pi}{2\sqrt{\bra{\underline{1}}H^2\ket{\underline{1}}-\bra{\underline{1}}H\ket{\underline{1}}^2}}=\frac{\pi}{2J_1},    \label{eqn:margolus}
\end{equation}
thereby giving an upper bound to the rate of perfect state transfer for any given scheme. The rate is also trivially bounded by $1/t_r\leq N/2t_0$, since it is impossible to store the states of more than $N$ qubits on an $N$-qubit chain. The following lemma allows us to prove a tighter bound of $1/t_r\leq N/4t_0$, although we conjecture a stronger condition; that there are no chains with $t_r<t_0$.

\begin{lemma}
If a set of eigenvalues is chosen to fulfill the perfect state transfer condition of Eqn.~(\ref{eqn:st_cond}), then a necessary and sufficient condition to perfectly achieve the rate $M/2t_0$ for integer $M$ is that all the $R_k$ for $k=0\ldots M-1$ should be equal, where
\begin{equation}
R_k=\sum_{n=1}^N\frac{(-1)^n}{B'(\lambda_n)},
\label{eqn:rate_cond}
\end{equation}
and the sum is restricted to those terms satisfying the condition
$$
\frac{t_0}{\pi}(\lambda_n-\lambda_1)\mod M=k.
$$
\label{lemma:rate}
\end{lemma}
\begin{proof}
Using Eqn.~(\ref{eqn:hoch}), we can write that
$$
\gamma_1(t)=\sum_{n=1}^N\frac{e^{-i\lambda_nt}}{(-1)^nB'(\lambda_n)\sum_{m=1}^N\frac{1}{(-1)^mB'(\lambda_m)}}.
$$
We will now demand that $\gamma_1(2mt_0/M)=0$ for integer $m=1\ldots M-1$ and $\gamma_1(2t_0)=1$. Using Eqn.~(\ref{eqn:rate_cond}), the terms in the sum can be split up to give
$$
\gamma_1(2mt_0/M)=\frac{\sum_{k=0}^{M-1}R_ke^{-i2\pi km/M}}{\sum_{k=0}^{M-1}R_k}=\delta_m.
$$
Using the Fourier Transform, these equations can be inverted to give that
$$
R_j=\frac{1}{M}\sum_{k=0}^{M-1}R_k,
$$
as required.
\end{proof}
Note that if $M$ is a composite number $M_1M_2$, then if the condition fails for either $M_1$ or $M_2$, it necessarily fails for $M$, since each of the $R_k$ for $M_1$ is given by a sum of $M_2$ of the $R_k$ for $M$. However, the converse is not true.
We conjecture that it is impossible to fulfill the condition of Lemma \ref{lemma:rate} for any $M>2$, although we only have a proof for $M>N/2$. In the case of $M=N$, since none of the terms $R_k$ can be 0 if they are to be equal, they must all be $R_k=(-1)^k/B'(\lambda_k)$. Given that $B'(\lambda_1)\neq -B'(\lambda_2)$ for $N>2$, it is immediate that it is impossible to perfectly achieve the rate $N/(2t_0)$. Similarly, for any $M>\lfloor N/2\rfloor+1$, the different $B'(\lambda_n)$ must be partitioned between all $R_k$, and it follows that at least two terms must be of the form $(-1)^n/B'(\lambda_n)$, and hence different.

\subsection{Sequential Quantum Storage}

While we may have no way to construct a perfect state transfer chain with high rate (retaining the perfect transfer property), it turns out that one can design a non-symmetric chain for which $\bra{\underline{1}}e^{-iH_1t_rn}\ket{\underline{1}}=0$ for $n=1\ldots 2t_0/t_r-1$, which means that the states can be read in sequentially at spin 1, and read out sequentially from the same spin at a time $2t_0$ later. Recall that
$$
\gamma_1(t)=\sum_{n=1}^N|\alpha_n|^2e^{-i\lambda_nt}.
$$
In order to get periodic revivals at the input spin, then even without symmetry, the chain must have eigenvalues such that the ratio of differences are rational (following an identical argument to that of Lemma \ref{lemma:symmetry_match}). So, we can select these to be the simplest conceivable case, $2t_0\lambda_n/\pi=-N+1,-N+3\ldots N-3,N-1$. In the previous subsection, we realized that if $t_r=2t_0/N$, then $|\alpha_n|^2=1/N$. To verify that this will work, we calculate
$$
\gamma_1(mt_r)=\frac{1}{N}\sum_{n=1}^Ne^{\frac{i2\pi m(-N-1+2n)}{2N}}=0,
$$
for any integer $m$ which is not a multiple of $N$. With the $\{|\alpha_n|^2\}$ and $\{\lambda_n\}$ specified, an iterative scheme for deciding the coupling strengths is evident; we simply calculate
$$
\bra{\underline{1}}H_1^m\ket{\underline{1}}=\sum_n|\alpha_n|^2\lambda_n^m
$$
for all integers $m$, and solve for the coupling strengths. It can be verified that a matrix with $B_n=0$ and
$$
J_n^2=\frac{n^2(N-n)(N+n)}{(2n-1)(2n+1)}
$$
has the desired properties, and $N$ single-qubit states can be sequentially stored in the chain, which must be the optimal solution in terms of density of storage.

Provided the bits are read-out in the reverse sequence to which they were stored, there are no undesired side-effects from the exchange of fermions, because no fermions are exchanged. On the other hand, if one reads them out in the same sequence as they were read-in, there is a controlled-phase gate enacted between every single qubit. This could be a very useful feature for generating GHZ states -- it would suffice to set all the states $\ket{\psi_n}=\ket{+}$, and the set of states that are output are equal, up to local unitaries, to an $N$-qubit GHZ state. By being selective over the order of removal, one can control which spins get entangled with which others\footnote{A state $\ket{\psi_n}$ becomes entangled, via a controlled-phase gate, with any state that is still on the chain at the time of removal, and which was placed on the chain after $\ket{\psi_n}$ was. All controlled phase gates that were enacted by previous removals remain unaffected.}. For storage of information, one can avoid the controlled-phase gates by again using a dual-rail encoding, as in Sec.~\ref{sec:sysinit}, thereby halving the number of qubits that can be stored.

This storage device can be used as the basis for a quantum computer, where all control is mediated through spin 1. Such a device is referred to as a Universal Quantum Interface \cite{UQI,Burgarth:07,sgs,Kay:08,Kay:08b}. The formulation of the control sequences for this scheme is particularly efficient \cite{Kay:09uqi}, although many other $XX$ model spin chains can also be used \cite{Kay:09uqi,Burgarth:uqi}.

\subsection{Other Fermionic Models} \label{sec:bogoliubov}

In this section, we have discussed the properties of multiple excitations on a spin chain. This was achieved by mapping the $XX$ model into a model of non-interacting fermions hopping on a nearest-neighbor chain, and proving transfer in these systems. However, the class of spin models that can be mapped into the same set of non-interacting fermions is much broader, which must mean that perfect state transfer is achievable in these systems, retaining the properties of transfer rate, independence of initial state etc.~\cite{paternostro:1,paternostro:2}.

We show this by considering the class of Hamiltonians
$$
H^{\gamma}=\half\sum_{n=1}^{N-1}J_n((1+\gamma_n)X_nX_{n+1}+(1-\gamma_n)Y_nY_{n+1})-\sum_{n=1}^NB_nZ_n.
$$
This includes our original $XX$ model as the special case $\gamma_n=0$ for all $n$, and the transverse Ising model, considered in \cite{paternostro:1,paternostro:2}, has $\gamma_n=1$. As before, the first step is to perform the Jordan-Wigner transformation,
\begin{eqnarray}
H^{\gamma}_{JW}&=&\sum_{n=1}^{N-1}J_n(a_n^\dagger a_{n+1}+a_{n+1}^\dagger a_n+\gamma a_n^\dagger a_{n+1}^\dagger+\gamma a_{n+1}a_n)	\nonumber\\
&&+2\sum_{n=1}^NB_na_n^\dagger a_n.	\nonumber
\end{eqnarray}
Following \cite{lieb}, this can be written in a general form of
$$
H^{\text{general}}=\sum_{n,m=1}^NA_{nm}a_n^\dagger a_m+\half B_{nm}(a_n^\dagger a_m^\dagger+a_ma_n)
$$
where $A$ is a real, symmetric, matrix ($A=A^T$) and $B$ is real and anti-symmetric ($B=-B^T$). We want to diagonalise this system by introducing new fermionic operators $c_k^\dagger$ such that
\begin{equation}
H^{\text{general}}=\sum_{k=1}^N\mu_kc_k^\dagger c_k,	\label{eqn:bogoliubov}
\end{equation}
and these new modes obey the canonical commutation relations
\begin{eqnarray}
\{c_n^\dagger,c_m^\dagger\}=\{c_n,c_m\}&=&0	\nonumber\\
\{c_n^\dagger,c_m\}&=&\delta_{n,m}.	\nonumber
\end{eqnarray}
Our intention is to expand these new modes in terms of the old modes,
$$
c_k^\dagger=\sum_ng_{kn}a_n^\dagger+h_{kn}a_n.
$$
In order for Eqn.~(\ref{eqn:bogoliubov}) to hold, it must be that
$$
\left[c_k,H^{\text{general}}\right]=\mu_kc_k,
$$
which gives a set of simultaneous equations
\begin{eqnarray}
\mu_kg_{kn}&=&\sum_mg_{km}A_{mn}-h_{km}B_{mn}	\nonumber\\
\mu_kh_{kn}&=&\sum_mg_{km}B_{mn}-h_{km}A_{mn}.	\nonumber
\end{eqnarray}
If we define new variables $\eta_{kn}=(g_{kn}+h_{kn})/\sqrt{2}$ and $\chi_{kn}=(g_{kn}-h_{kn})\sqrt{2}$, then these equations can be succinctly expressed as a matrix equation
\begin{eqnarray}
\mu_k\left(\begin{array}{c} \eta_k \\ \chi_k \end{array}\right)&=&\left(\begin{array}{cc}
0 & A+B \\
A-B & 0 \end{array}\right)\left(\begin{array}{c} \eta_k \\ \chi_k \end{array}\right)	\nonumber
\end{eqnarray}
so the combined vector $\left(\begin{array}{c} \eta_k \\ \chi_k \end{array}\right)$ is an eigenvector of a $2N\times 2N$ Hermitian matrix. We can write that
\begin{eqnarray}
\{c_n^\dagger,c_m^\dagger\}&=&\sum_kg_{nk}h_{mk}+g_{mk}h_{nk}	\nonumber\\
&=&\sum_k\eta_{nk}\eta_{mk}-\chi_{nk}\chi_{mk}	\nonumber
\end{eqnarray}
and, similarly,
\begin{eqnarray}
\{c_n^\dagger,c_m\}&=&\sum_kg_{nk}g_{mk}+h_{mk}h_{nk}	\nonumber\\
&=&\sum_k\eta_{nk}\eta_{mk}+\chi_{nk}\chi_{mk}.	\nonumber
\end{eqnarray}
However, these correspond directly to the orthogonality relations for the eigenvectors, and hence the canonical commutation relations hold.

Let us take, for instance, the special case of $\gamma=1$, and a nearest-neighbor interaction. This corresponds to the transverse Ising model, as considered in \cite{paternostro:1,paternostro:2}. If we index the basis elements of
$$
M=\left(\begin{array}{cc}
0 & A+B \\
A-B & 0 \end{array}\right)
$$
by $\ket{\underline{n}}$ for $n=1\ldots 2N$, then we have
\begin{eqnarray}
M\ket{\underline{N+n}}&=&B_n\ket{\underline{n}}+2J_{n-1}\ket{\underline{n-1}}	\nonumber\\
M\ket{\underline{n}}&=&B_n\ket{\underline{N+n}}+2J_{n}\ket{\underline{N+1+n}}	\nonumber
\end{eqnarray}
This is exactly the nearest-neighbor hopping Hamiltonian of $2N$ qubits. Hence, we can take any perfect state transfer solution for $2N$ qubits with 0 magnetic field and coupling strengths $K_n$ and identify $B_{n}=K_{2n-1}$ and $2J_{n}=K_{2n}$, and this is guaranteed to give perfect state transfer in this new model. For instance, if we start from a state $\ket{\underline{0}}$, defined such that $c_k\ket{\underline{0}}=0$ for all $k$, then we can create a state $\alpha\ket{\underline{0}}+\beta a_1^\dagger\ket{\underline{0}}$. As before the $\ket{\underline{0}}$ component is an eigenstate, and remains unchanged. The state $a_1^\dagger\ket{\underline{0}}$ corresponds to $\ket{\underline{1}}+\ket{\underline{N+1}}$ in $M$, and gets transferred to $\ket{\underline{N}}+\ket{\underline{2N}}$, which therefore outputs $a_N^\dagger\ket{\underline{0}}$, so the state is perfectly transferred.

\section{Noise Tolerance} \label{sec:4}

In a practical situation, no matter how well our system is engineered, it will never be perfectly protected from an environment (although one might hope that restricting access to the chain to just the ends would allow us to protect the bulk of the chain to a large degree). Thus, it is worthwhile to make some consideration of errors; all this work would be for nothing if the moment a small degree of noise is introduced, the quality of transfer were to drop radically.

The presence of static defects, i.e.~manufacturing errors, is a well-studied problem, resulting in the phenomenon of Anderson localization \cite{anderson} for sufficiently large defects. Small, perturbative, errors lead to perturbative errors in the arrival fidelities \cite{simone}. Instead, we consider dynamical noise processes, singling out two independent models, local dephasing and local bit flips. The individual behaviors of these can be expected to be potentially very different, although we do not expect either to be excessively harmful. In the following subsection, we will discuss the case of dephasing noise. The case of spin flips is closely related to that of transfer rates and the initial state of the system (Secs.~\ref{sec:transrate} and \ref{sec:sysinit}), where we learned that the effect of the bit flip is very small if the transferred state is not present on that particular spin at that particular time. Unfortunately, rigorous calculations for bit-flip noise are extremely complicated, and not informative, as they necessarily involve the complete $2^N$-dimensional Hilbert space of the system. More detailed examinations have been made in \cite{Burgarth,Cai}.

Ultimately, in the implementation of a quantum computation, the computational qubits are encoded into error correcting codes. If one were to send only one physical qubit along the state transfer chain at a time, the end result would be independent errors on the physical spins, which error correcting codes, and the theory of fault-tolerance \cite{nielsen}, are well tuned to handling. There would be no memory effect since we know how to transmit through a chain without any effect due to the initial state of the chain (Sec.~\ref{sec:sysinit}), which is the only way any initially independent noise could potentially build up any correlation.

Of further interest is the fact that it has recently been shown that for general spin networks, where perfect state transfer is impossible in a coherent way, the transfer ability can be vastly {\em enhanced} by the presence of noise, and it seems that this may have significance for some biological processes \cite{plenio:bio}.

\subsection{Dephasing Noise}

Dephasing noise, the random application of $Z$ errors with probability $p$, independently on each lattice site, could be expected to be quite harmful since the relative phases that build up between sites during the state transfer are critical to the perfect transfer. For the sake of analyticity, we shall restrict to a much simpler model where dephasing noise is only applied at one time $t$ ($0\leq t\leq t_0$). Thus, our initial state
$$
\ket{\Psi}=\ket{\psi}\ket{0}^{\otimes N-1}
$$
evolves to
$$
\rho_{\text{out}}=\sum_{x\in\{0,1\}^N}(1-p)^{N-w_x}p^{w_x}\tilde{Z_x}\proj{\Psi}\tilde{Z_x}^\dagger,
$$
where $w_x$ denotes the Hamming weight of the bit string $x$, used to enumerate the set of spins on which the $Z$ rotations have been applied. The Hamiltonian evolution has been incorporated into
$$
\tilde{Z_x}=e^{-iH(t_0-t)}Z_xe^{-iHt}.
$$
If we denote the $n^{th}$ bit of $x$ by $x_n$, then
$$
\tilde{Z_x}\ket{\Psi}=e^{-iHt_0}\ket{\Psi}-2\beta\sum_{x_n=1}\gamma_ne^{-iH(t_0-t)}\ket{n}.
$$
Our aim is to determine how well the state is transferred, by calculating the fidelity,
$$
F=\bra{\psi}\Tr_{1\ldots N-1}\left(\rho_{\text{out}}\right)\ket{\psi},
$$
and then averaging over all possible input states $\ket{\psi}$ to give $\langle F\rangle$. Ultimately, this yields that
$$
\langle F\rangle=1-\frac{2p(2-p)}{3}+\frac{2p(1-p)}{3}\sum_{n=1}^N|\gamma_n(t)|^4.
$$
Trivial bounds can be applied for the $|\gamma_n(t)|^4$ to give
$$
1-\frac{2p(2-p)}{3}+\frac{2p(1-p)}{3N}\leq\langle F\rangle\leq1-\frac{2p}{3},
$$
where the upper bound corresponds to the case of simple storage without transfer. We note that at any typical time during the transfer process, the state will be significantly spread out across the chain, and the fidelity is expected to be closer to the bottom end of the range, unless one uses a tactic such as encoding the state into a wavepacket which remains well localized during the transfer \cite{osborne04}.

\subsection{Independent Baths}

Another reasonable noise model, discussed in \cite{Burgarth}, is where each spin of the chain is coupled to a localized set of spins by a similar interaction. While we will consider the aspects of the proof necessary for the class of chains that we are interested in, note that the proof of \cite{Burgarth} is remarkably general, and essentially any Hamiltonian that splits into excitation subspaces will exhibit almost identical properties.

\begin{figure}
\begin{center}
\includegraphics[width=0.45\textwidth]{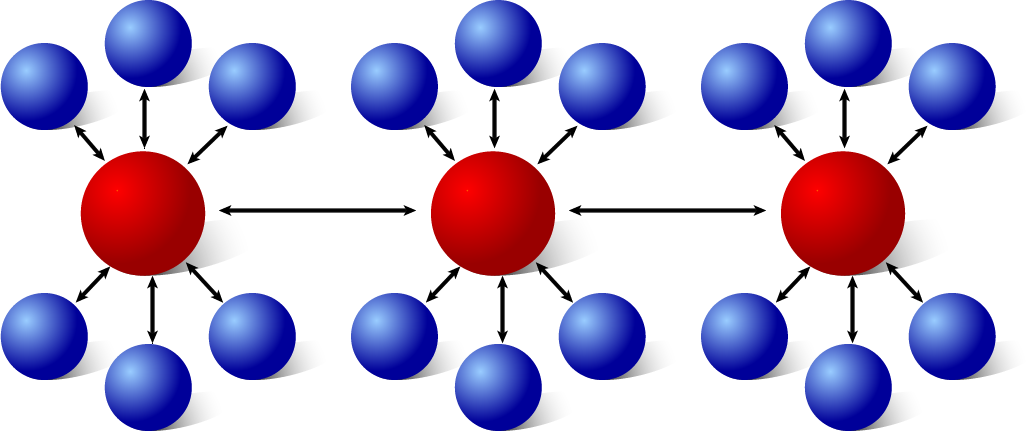}
\end{center}
\caption{The system spins (large), intended for perfect state transfer, are coupled to independent baths (small).} \label{fig:baths}
\end{figure}

Our aim, as previously, is to transfer states through a specific chain,
$$
H_S=\half\sum_{n=1}^{N-1}J_n(X_nX_{n+1}+Y_nY_{n+1})-\sum_{n=1}^NB_nZ_n,
$$
but each of the spins $n$ will also be coupled to an independent bath composed of $M_n$ spins as depicted in Fig.~\ref{fig:baths},
$$
H_{B,n}=\half\sum_{m=1}^{M_n}g^n_m(X_nX_{(n,m)}+Y_nY_{(n,m)}),
$$
yielding a total Hamiltonian
$$
H=H_S+\sum_{n=1}^NH_{B,n}.
$$
It is useful to observe that if an excitation starts on one of the system spins $n$, then while it can hop arbitrarily to other system spins, if it appears on a bath spin, it can only appear in the states
$$
\frac{1}{G_n}\sum_{m=1}^{M_n}g^n_m\ket{\underline{(n,m)}},
$$
where $G_n^2=\sum_{m=1}^{M_n}(g^n_m)^2$ \cite{Kay:2005b}. This means that a lot of the complexity of the bath spins can be eliminated if we restrict to the single excitation subspace, with
$$
H_{B,n}'=\half G_n(X_nX_{n'}+Y_nY_{n'}),
$$
where $n'$ is a single effective spin which acts as the bath for spin $n$. We now assume that all the $G_n$ are equal\footnote{This makes no assumptions regarding the individual $g^n_m$, so one intriguing possibility would be to deliberately introduce one extra spin at each site whose $g^n$ can be controlled in order to tune the $G_n$.}. The eigenvectors of the system in the absence of the bath can be written as
$$
\ket{\lambda_m}=\sum_{n=1}^N\lambda_{m,n}\ket{\underline{n}},
$$
allowing us to make an ansatz for the eigenstates of $H'=H_S+\sum_nH_{B,n}'$:
$$
\ket{\Lambda_m^k}=\frac{1}{\sqrt{2}}\sum_{n=1}^N\lambda_{m,n}(\ket{\underline{n}}+(-1)^k\ket{\underline{n'}}).
$$
These satisfy
$$
\bra{\Lambda_m^k}H'\ket{\Lambda_n^l}=\delta_{m,n}((-1)^kG\delta_{k,l}+\half\lambda_n),
$$
so although $H'$ is not instantly diagonalized, it is reduced to $2\times 2$ sub-blocks which are diagonalizable to give the eigenvectors and energies.
$$
E^k_n=\half(\lambda_n+(-1)^k\sqrt{4G^2+\lambda_n^2})
$$
Having solved for these, we can now turn our attention to how well this system implements state transfer. The accuracy of transfer of the system in the absence of the baths at time $t$ is given by $\gamma_N(t)\equiv\gamma_1^*(t_0-t)$. In the presence of the baths, one can show two results. Firstly, in the weak coupling limit, the first correction to $\gamma_N'(t)$ is only of order $G^2$. In the strong coupling limit, one finds that
$$
\gamma_N'(t)=\cos(Gt)\gamma_N(t/2),
$$
showing that the transfer takes a factor of 2 longer, and is modulated by a fast oscillating term. This means that extremely high fidelity, or even perfect, state transfer is still possible!

\section{Perfect State Transfer in Higher Spatial Dimensions} \label{sec:highD}

With the solutions for perfect state transfer in 1D, it is straightforward to construct schemes for state transfer in larger spatial dimensions. Let us consider two different coupling schemes, one for $N$ qubits with nearest-neighbor couplings $\{J_n\}$ and magnetic fields $B_n=0$, and a second one for $M$ qubits, nearest-neighbor couplings $\{K_n\}$ and magnetic fields $B_n=0$. Both have the same transfer time $t_0$. Now consider an $N\times M$ lattice of qubits, indexed by their positions $(x,y)$. A Hamiltonian of the form
\begin{eqnarray}
H_{2D}&=&\half\sum_{i=1}^{N-1}\sum_{j=1}^MJ_n(X_{(i,j)}X_{(i+1,j)}+Y_{(i,j)}Y_{(i+1,j)}) \nonumber\\
&&+\half\sum_{i=1}^{N}\sum_{j=1}^{M-1}K_n(X_{(i,j)}X_{(i,j+1)}+Y_{(i,j)}Y_{(i,j+1)}) \nonumber
\end{eqnarray}
performs state transfer between any diagonally opposite points $(i,j)$ and $(N+1-i,M+1-j)$; the two dimensions behave independently, and the eigenvectors of the first excitation subspace factorize into the form
$$
\ket{\lambda^{N,M}_{n,m}}=\sum_{i=1}^N\sum_{j=1}^M\lambda_{n,i}^N\lambda_{m,j}^M\ket{\underline{(i,j)}}
$$
where the eigenvectors for the two schemes are
$$
\begin{array}{c}
\ket{\lambda^N_n}=\sum_{i=1}^N\lambda_{n,i}^N\ket{\underline{i}} \\
\ket{\lambda^M_m}=\sum_{j=1}^M\lambda_{m,j}^M\ket{\underline{j}}.
\end{array}
$$
The eigenvalues satisfy
$$
\lambda^{N,M}_{n,m}=\lambda^N_n+\lambda^M_m.
$$
The Jordan-Wigner transformation does not apply to $H_{2D}$, so perfect state transfer is not guaranteed in the presence of multiple excitations for the $XX$ model. Thus, the question of whether full state mirroring can occur on 2D lattices still remains open, although if one assumes certain symmetry conditions on the coupling strengths, it can be shown that mirroring is impossible \cite{Kay:square}. On the other hand, were one to use a bosonic model, such as that of coupled Harmonic oscillators, this transfer must work in all excitation subspaces.

The described technique is not limited to combining two 1D chains into a 2D Hamiltonian. The introduction of the third chain allows the construction to be extended to 3D, and so on {\em ad infinitum}. The special case of applying this protocol with the chain of two qubits generates a uniformly coupled hypercube, as studied in \cite{Kay:2004c}, and is the basis (or special case) of several of the perfect transfer network variants \cite{facer, simone2}.

\section{Transformations and Manipulations} \label{sec:6}

The motivation that is often given for the study of perfect state transfer is that these devices might be realized in quantum computers, simplifying some tasks. Of course, state transfer is not the only task that one might want to realize. Having laid the ground work of perfect state transfer, however, it is vastly simpler to discover how these other `gadgets' might be realized, as we will momentarily examine.

Nevertheless, the engineering constraints on these modulated chains are quite stringent, and may not simplify anything at all. In that case, is all that we have learned wasted? Certainly not! It turns out that the study of perfect state transfer is a fantastically useful constructive technique to design all sorts of schemes based on Hamiltonian evolution, whether these be constructive or destructive effects. Perhaps one wants to show how evolution of a Hamiltonian can implement a quantum computation, or perhaps one wants to design a perturbation or interaction with an environment that disrupts a coherent effect (thereby giving bounds on the range of usefulness of that effect) \cite{Kay:09}. The details of these individual calculations will not concern us here, but in Sec.~\ref{sec:universal}, we will describe the basic method that one uses.

\subsection{Interferometry and W-state Preparation}

We can view the action of the Hamiltonian $H_1$ ($H$ acting on the first excitation subspace) as a discrete hopping along a chain of $N$ orthogonal states $\{\ket{\underline{n}}\}$;
$$
H_1\ket{\underline{n}}=\left\{\begin{array}{cc}
B_1\ket{\underline{1}}+J_1\ket{\underline{2}}  & n=1   \\
J_{N-1}\ket{\underline{N-1}}+B_N\ket{\underline{N}}    & n=N   \\
J_{n-1}\ket{\underline{n-1}}+B_n\ket{\underline{n}}+J_n\ket{\underline{n+1}}   & \text{otherwise}
\end{array}\right.
$$
By selecting the $J_n$ and $B_n$ so that they correspond to a perfect state transfer setting, we are guaranteed that the Hamiltonian evolves states $\ket{\underline{n}}$ into $\ket{\underline{N+1-n}}$ in time $t_0$, no matter what these states are. As such, we can construct a plethora of different protocols.

\begin{figure}
\begin{center}
\includegraphics[width=0.45\textwidth]{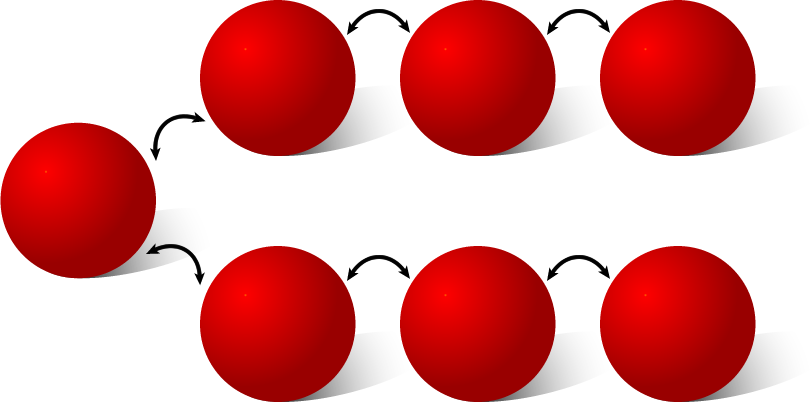}
\end{center}
\caption{Two parallel chains, joined at the start, can be used for entanglement generation.} \label{fig:ent_gen}
\vspace{-0.5cm}
\end{figure}

The simplest of these is involved in producing a beam-splitter action. We introduce a new system composed of two parallel chains ($a$ and $b$) of spins, joined at one end, spin 1, as depicted in Fig.~\ref{fig:ent_gen}. The Hamiltonian of this system can be written as
\begin{eqnarray}
&\sum_{n=2}^{N-1}\sum_{i\in\{a,b\}}B_nZ_{n,i}+J_n(X_{n,i}X_{n+1,i}+Y_{n,i}Y_{n+1,i})& \nonumber\\
&+B_1Z_1+\frac{J_1}{\sqrt{2}}\sum_{i\in\{a,b\}}(X_1X_{2,i}+Y_1Y_{2,i}).&  \nonumber
\end{eqnarray}
Using the position basis as before, i.e.~$\ket{\underline{2,a}}$ indicates the presence of an excitation on spin $2a$ and nowhere else, we can write a set of orthogonal states to be $\ket{\tilde 1}=\ket{\underline{1}}$ and
$$
\ket{\tilde n}=\frac{1}{\sqrt{2}}(\ket{\underline{n,a}}+\ket{\underline{n,b}}).
$$
One can then calculate the effect of the new Hamiltonian on these states, and see that it maps to the action of $H_1$. As a result, if we start with a state $\ket{1}$, and allow it to evolve for a time $t_0$, we end up with the state $\ket{\tilde N}$, which is a maximally entangled state between two spins. This is easily generalized to creating a W-state on $M$ qubits simply by introducing $M$ parallel chains, and joining them all to spin 1 with the coupling strength $J_1/\sqrt{M}$. These manipulations were originally expounded in \cite{Kay:2005b,plenio:0,plenio:05}, although this idea has subsequently resurfaced a number of times \cite{yang:05,damico}.

For the special case of $N=2$, but arbitrary $M$, there is a neat trick that could be quite helpful. We can readily enumerate the other eigenvectors of the system,
$$
\ket{W_k}=\sum_{j\in\{a\ldots m\}}e^{2\pi ikj/M}\ket{\underline{2,j}},
$$
and further note that the output from the evolution ($\ket{W_0}$) can be converted into any of these via the application of phase gates. This means that once the $t_0$ evolution has completed, application of local unitaries can trap the entangled state on the spins 2. Indeed, this is the basis of the routing networks in \cite{in_prep}.

In \cite{Kay:2005b}, these ideas were taken further, constructing networks whereby unitary rotations can be performed while states are transmitted along the spin chains. For example, an interferometer can be constructed using the network depicted in Fig.~\ref{fig:X}.

\begin{figure}[!tb]
\begin{center}
\includegraphics[width=0.45\textwidth]{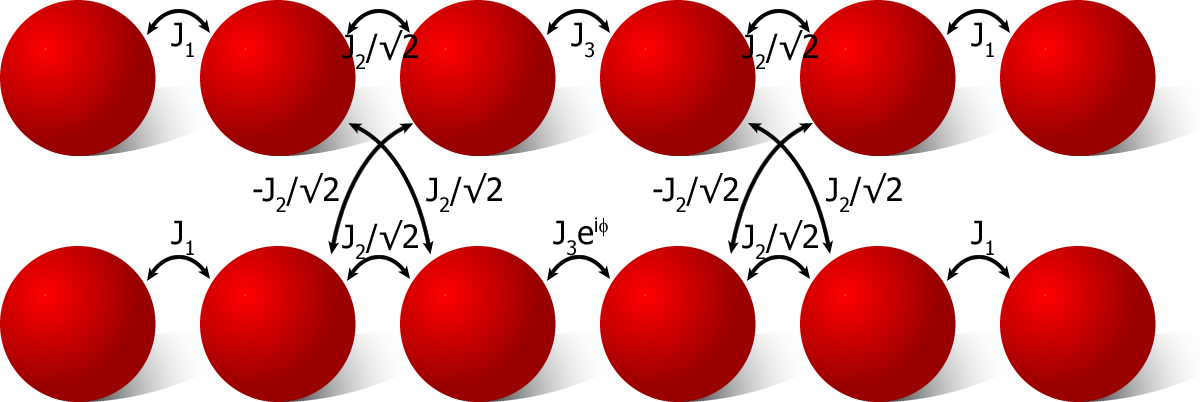}
\end{center}
\caption{A spin chain network ($N=6$) that acts like an interferometer, detecting the phase $\phi$ of the central coupling via the probability of which of the output ports an input excitation arrives at.}
\label{fig:X}
\vspace{-0.5cm}
\end{figure}

\subsection{An Alternate Solution for Entanglement Generation}

Let us consider a chain of odd length, $2N+1$, that is capable of perfect state transfer, $\ket{\underline{n}}\rightarrow\ket{\underline{2N+2-n}}$. We can rewrite this chain as two effective chains given by states in the symmetric subspace, $\{\ket{n_+}=(\ket{\underline{n}}+\ket{\underline{2N+2-n}})/\sqrt{2}\}$ and $\ket{\underline{N+1}}$, and the antisymmetric subspace $\{\ket{n_-}=(\ket{\underline{n}}-\ket{\underline{2N+2-n}})/\sqrt{2}\}$. The first of these chains is described by the interaction
\begin{eqnarray}
H_+&=&\sum_{n=1}^{N-1}J_n(\ket{n_+}\bra{{n+1}_+}+\ket{{n+1}_+}\bra{n_+})	\nonumber\\
&&+\sum_{n=1}^{N+1}B_n\proj{n_+}	\nonumber\\
&&+\sqrt{2}J_{N}(\ket{N_+}\bra{{N+1}_+}+\ket{{N+1}_+}\bra{N_+})	\nonumber
\end{eqnarray}
and the second is
\begin{eqnarray}
H_-&=&\sum_{n=1}^{N-1}J_n(\ket{n_-}\bra{{n+1}_-}+\ket{{n+1}_-}\bra{n_-})	\nonumber\\
&&+\sum_{n=1}^{N}B_n\proj{n_-}.	\nonumber
\end{eqnarray}
Furthermore, we know what evolutions these must give in the perfect transfer time; $\ket{n_+}\rightarrow\ket{n_+}$ and $\ket{n_-}\rightarrow-\ket{n_-}$. Now let us consider altering the chain slightly, $\tilde J_N=\sqrt{2}\cos\theta J_N$, $\tilde J_{N+1}=\sqrt{2}\sin\theta J_{N+1}$ (recall that $J_{N+1}=J_N$). We can now apply exactly the same analysis as before, except now we have that $\{\ket{\tilde n_+}=(\cos\theta\ket{n}+\sin\theta\ket{2N+2-n})\}$ and $\{\ket{\tilde n_-}=(\sin\theta\ket{n}-\cos\theta\ket{2N+2-n})\}$, leaving $H_\pm$ unchanged, and consequently leaving their action unchanged. Hence, were we to start our system in the state $\ket{\underline{1}}$, this can be represented as
$$
\ket{\underline{1}}=\cos\theta\ket{\tilde 1_+}+\sin\theta\ket{\tilde 1_-}
$$
and in the perfect state transfer time, this evolves to
$$
\cos\theta\ket{\tilde 1_+}-\sin\theta\ket{\tilde 1_-}=\cos(2\theta)\ket{\underline{1}}+\sin(2\theta)\ket{2N+1}.
$$
Thus, by selecting $\cos(2\theta)=1/\sqrt{2}$, the chain creates the maximally entangled state $(\ket{\underline{1}}+\ket{\underline{2N+1}})/\sqrt{2}$ between the furthermost points of the chain \cite{kwek}.

\subsection{Signal Amplification and GHZ-state Preparation} \label{sec:amp}

Another simple protocol that might be useful in a quantum computer is to amplify a signal before measurement \cite{kay-2006b}. The idea would be to take an unknown state $\ket{\psi}=\alpha\ket{0}+\beta\ket{1}$ and produce a GHZ state of the form $\alpha\ket{0}^{\otimes N}+\beta\ket{1}^{\otimes N}$, yielding a net magnetization that is more readily detected by a measuring device. The sensible approach here would seem to be to demand that the state $\ket{0}^{\otimes N}$ should be an eigenstate of our Hamiltonian $H_{\text{amp}}$, and that the states
$$
\ket{\tilde n}=\ket{1}^{\otimes n}\ket{0}^{\otimes N-n}
$$
should form the steps that we need. One can check that a Hamiltonian term $K_m=\half X_m(\identity-Z_{m-1}Z_{m+1})$ satisfies $K_m\ket{0}^{\otimes N}=0$ and
$$
K_m\ket{\tilde n}=\left\{\begin{array}{cc}
\ket{\widetilde{n+1}}   & m-1=n \\
\ket{\widetilde{n-1}}   & m=n \\
0                   & \text{otherwise}
\end{array}\right.
$$
($Z_{N+1}$ is just taken to be $1$), so $K_{m+1}$ replaces the role of the term $\half(X_mX_{m+1}+Y_mY_{m+1})$. Similarly, $Z_mZ_{m+1}$ replaces $Z_m$. Thus, we can write down the required Hamiltonian as
$$
H_{\text{amp}}=\sum_{n=1}^{N-1}J_nK_{n+1}+\sum_{n=1}^NB_nZ_nZ_{n+1}.
$$
Fig.~(\ref{fig:sig_strength}) gives a plot of how well the signal gets amplified as a function of time using the set of couplings from the analytic solution of perfect state transfer.

\begin{figure}
\begin{center}
\includegraphics[width=0.45\textwidth]{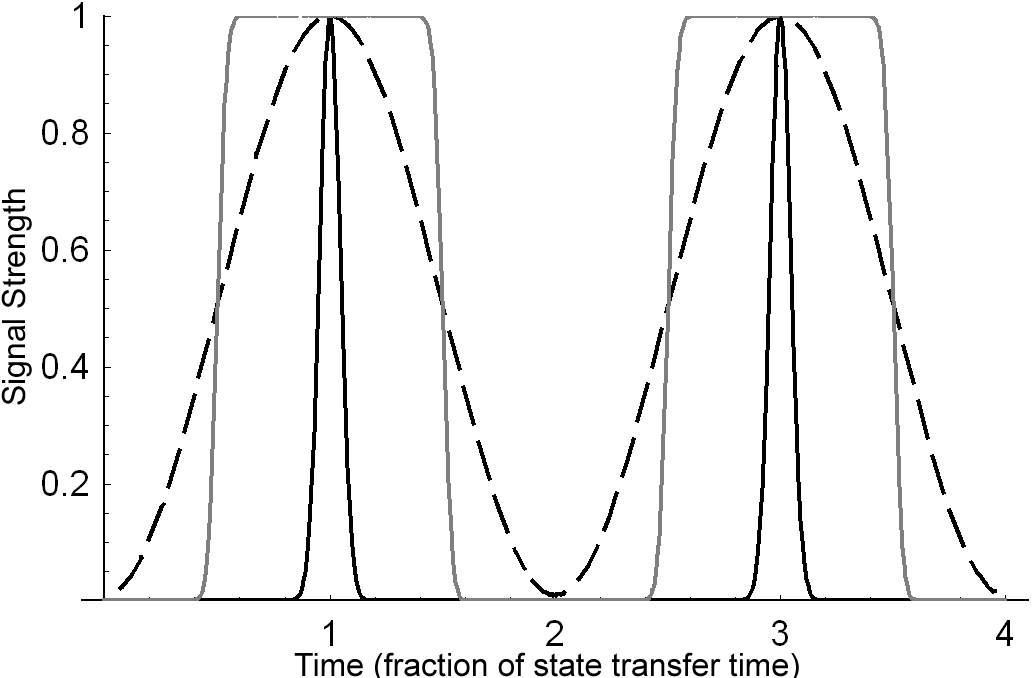}
\vspace{-0.2cm}
\end{center}
\caption[Signal Amplification]{Signal amplification for a system of $N=100$ qubits using $J_n=\sqrt{n(N-n)}$. Perfect amplification occurs at times $(2n+1)t_0$ and revival of the initial state occurs at times $2nt_0$. The solid line indicates the probability of the signal (initially $\ket{\tilde 1}$) being amplified to the desired final state $\ket{\tilde N}$. The dashed line shows the average signal strength (as a fraction of the maximum strength), and the gray line plots the probability that more than half of the spins have been flipped. Both the dashed and gray lines are largely independent of the number of qubits in the system.} \label{fig:sig_strength}
\vspace{-0.5cm}
\end{figure}

It turns out that one can show a unitary transformation between $H$ and $H_{\text{amp}}$ \cite{kay-2006b}, not just their first excitation subspaces. For example, they have an identical subspace structure, where the required commutation property of $H_{\text{amp}}$ is
$$
\left[H_{\text{amp}},\sum_{n=1}^{N}Z_{n}Z_{n+1}\right],
$$
so the number of boundaries between sets of $\ket{0}$s and $\ket{1}$s on a chain is preserved.
As a result, one can apply the results about multiple excitations to this system as well. One such result is that a single excitation on a spin $n\geq 2$ corresponds to a pair of excitations $\ket{\underline{n-1,n}}$ on $H$, and gets perfectly mirrored to spin $N+2-n$. Therefore, this chain is capable of performing perfect state transfer. Moreover, because we are using the effective two-excitation subspace, this behaves like the encoding of Sec.~{\ref{sec:sysinit}), which does not require preparation of the initial state of the system and yet, in this case, does not require any encoding/decoding of the quantum state to be transferred.

\subsection{Universal Quantum Computation} \label{sec:universal}

Ultimately, we may be interested in how far one can go using these simple constructions. Almost by definition, a Hamiltonian evolution cannot be more powerful than a quantum computer\footnote{The Hamiltonian evolution can be split into a number of small steps, and a Trotter decomposition allows each step to be written as $n$-qubit gates for an $n$-body Hamiltonian, and these can be further decomposed into a standard universal gate set within the circuit model of Quantum Computation \cite{nielsen}.}. So, is its power the same, or less? For now, we will not be concerned with constructing a Hamiltonian which is local with regard to any particular spatial distribution; we are merely demonstrating the general technique. Evidently, one can write down a quantum computation as a sequence of unitaries $U_1$ to $U_{N-1}$ which should be applied to the initial state $\ket{\psi_{\text{in}}}$. However, it is not necessarily true that the state after $n$ of these unitaries has been applied is orthogonal to the state at any other time step $m$, although this is required for the Hamiltonian evolution. To circumvent this problem, we introduce an ancillary system known as the clock. This clock will have $N$ orthogonal states denoted by $\ket{n}_C$. The sequence of states that we wish to step through are now
$$
\ket{\tilde n}=\ket{n}_C\otimes\left(U_{n-1}\ldots U_1\ket{\psi_{\text{in}}}\right)
$$
so that by initializing the clock in the state $\ket{1}_C$, in time $t_0$ it will arrive at the state $\ket{N}_C$, and the computation will have been implemented. It remains to write down a suitable Hamiltonian,
\begin{eqnarray}
H_{\text{comp}}&=&\sum_{n=1}^{N-1}J_n\ket{n+1}\bra{n}_C\otimes U_n+h.c. \nonumber\\
&&+\sum_{n=1}^NB_n\proj{n}_C\otimes\identity,   \nonumber
\end{eqnarray}
where $h.c.$ denotes the Hermitian conjugate. Thus, Hamiltonian evolution is as powerful as quantum computation. Such a construction, using perfect transfer, was originally introduced by Peres \cite{peres} as a modification of Feynman's scheme for Hamiltonian-based quantum computation \cite{feynman}. There are a number of ways in which this can be extended. For example, one can impose further structure on the Hamiltonian, such as a 1D translationally-invariant and $SU(2)$-invariant structure, and still recover that arbitrary quantum computations can be implemented \cite{Kay:08,karl,Kay:2009x}. These constructions can also make use of cellular automata for the implementation of the computation, which means that the gates to be applied are just repetitions of the same basic gate, and the computation is entirely determined by the initial state of the system. Alternatively, one can use these Hamiltonians to understand how difficult it is to find the ground state energies \cite{Kempe,gottesman,KSV02a,oliveira-2005,Kay:08}, or to discuss the entanglement properties of ground states \cite{irani}.

\section{Conclusions} \label{sec:conc}

In this article we have reviewed to concept of perfect state transfer, and derived the necessary and sufficient conditions for a nearest-neighbor $XX$ or Heisenberg Hamiltonian. In the case of $XX$ Hamiltonians, we have explored the properties of the transfer in the presence of many excitations which, amazingly, do not hinder the effect. We have also demonstrated how a firm grasp of perfect state transfer forms the basis of a constructive technique with a wide range of applicability, underlining the importance of this protocol, even if perfect state transfer chains are never realized in a laboratory.


\begin{thebibliography}{99}
\expandafter\ifx\csname natexlab\endcsname\relax\def\natexlab#1{#1}\fi
\expandafter\ifx\csname bibnamefont\endcsname\relax
  \def\bibnamefont#1{#1}\fi
\expandafter\ifx\csname bibfnamefont\endcsname\relax
  \def\bibfnamefont#1{#1}\fi
\expandafter\ifx\csname citenamefont\endcsname\relax
  \def\citenamefont#1{#1}\fi
\expandafter\ifx\csname url\endcsname\relax
  \def\url#1{\texttt{#1}}\fi
\expandafter\ifx\csname urlprefix\endcsname\relax\def\urlprefix{URL }\fi
\providecommand{\bibinfo}[2]{#2}
\providecommand{\eprint}[2][]{\url{#2}}

\bibitem[{\citenamefont{Bose}(2003)}]{Bos03}
\bibinfo{author}{\bibfnamefont{S.}~\bibnamefont{Bose}}, \bibinfo{journal}{Phys.
  Rev. Lett.} \textbf{\bibinfo{volume}{91}}, \bibinfo{pages}{207901}
  (\bibinfo{year}{2003}).

\bibitem[{\citenamefont{Christandl et~al.}(2005)\citenamefont{Christandl,
  Datta, Dorlas, Ekert, Kay, and A.}}]{Kay:2004c}
\bibinfo{author}{\bibfnamefont{M.}~\bibnamefont{Christandl}},
  \bibinfo{author}{\bibfnamefont{N.}~\bibnamefont{Datta}},
  \bibinfo{author}{\bibfnamefont{T.}~\bibnamefont{Dorlas}},
  \bibinfo{author}{\bibfnamefont{A.}~\bibnamefont{Ekert}},
  \bibinfo{author}{\bibfnamefont{A.}~\bibnamefont{Kay}}, \bibnamefont{and}
  \bibinfo{author}{\bibnamefont{A.~Landahl}}, \bibinfo{journal}{Phys. Rev. A}
  \textbf{\bibinfo{volume}{71}}, \bibinfo{pages}{032312}
  (\bibinfo{year}{2005}).

\bibitem{Bose_review} S.~Bose, Contemp.~Phys.~{\bf 48}, 13 (2007).

\bibitem[{\citenamefont{Christandl et~al.}(2004)\citenamefont{Christandl,
  Datta, Ekert, and Landahl}}]{Christandl}
\bibinfo{author}{\bibfnamefont{M.}~\bibnamefont{Christandl}},
  \bibinfo{author}{\bibfnamefont{N.}~\bibnamefont{Datta}},
  \bibinfo{author}{\bibfnamefont{A.}~\bibnamefont{Ekert}}, \bibnamefont{and}
  \bibinfo{author}{\bibfnamefont{A.~J.} \bibnamefont{Landahl}},
  \bibinfo{journal}{Phys. Rev. Lett.} \textbf{\bibinfo{volume}{92}},
  \bibinfo{pages}{187902} (\bibinfo{year}{2004}).

\bibitem{lambrop} G.~M.~Nikolopoulos, D.~Petrosyan and P.~Lambropoulos, Europhys.~Lett.~{\bf 65}, 297 (2004); J. Phys.: Condens.~Matter {\bf 16}, 4991 (2004).

\bibitem{li} L.~Ying, S.~Zhi and S.~Chang-Pu, Commun.~Theor.~Phys.~{\bf 48} 445 (2007).

\bibitem{plenio:0} M.~B.~Plenio, J.~Hartley and J.~Eisert, New J.~Phys.~{\bf 6}, 36 (2004).

\bibitem{osborne04}
T.~J.~Osborne and N.~Linden, Phys.~Rev.~A \textbf{69}, 052315
(2004).

\bibitem{Burgarth:2} D.~Burgarth and S.~Bose, Phys.~Rev.~A {\bf 71}, 052315 (2005).

\bibitem{haselgrove04}
H.~L.~Haselgrove, Phys.~Rev.~A {\bf 72}, 062326 (2005).

\bibitem{networks} A.~W\'ojcik, T.~Luczak, P.~Kurzynski, A.~Grudka, T.~Gdala and M.~Bednarska, Phys.~Rev.~A {\bf 72}, 034303 (2005).

\bibitem{dark} T.~Ohshima, A.~Ekert, D.~K.~L.~Oi, D.~Kaslizowski and L.~C.~Kwek, quant-ph/0702019;
K.~Eckert, O.~Romero-Isart and A.~Sanpera, New J.~Phys.~{\bf 9}, 155 (2007).


\bibitem{simone3} N.~Saxena, S.~Severini and I.~E.~Shparlinski,  Int.~J.~Quantum Inf.~{\bf 5}, 417 (2007).

\bibitem{facer} C.~Facer, J.~Twamley and J.~D.~Cresser, Phys.~Rev.~A {\bf 77}, 012334 (2008).

\bibitem{simone2} A.~Bernasconi, C.~Godsil and S.~Severini, Phys.~Rev.~A {\bf 78}, 052320 (2008).

\bibitem{feder} D.~L.~Feder, Phys.~Rev.~Lett.~{\bf 97}, 180502 (2006).

\bibitem{topology} V.~Kostak, G.~M.~Nikolopoulos, I.~Jex, Phys.~Rev.~A {\bf 75},
042319 (2007). 

\bibitem{ahmadi} O.~Ahmadi, N.~Alon, I.~F.~Blake and I.~E.~Shparlinski, Linear Algebra Appl.~{\bf 430}, 547 (2009).

\bibitem{in_prep} P.~J.~Pemberton-Ross and A.~Kay, {\em in preparation}.

\bibitem{paternostro:1} C.~Di Franco, M.~Paternostro, D.~I.~Tsomokos, and S.~F.~Huelga, Phys.~Rev.~A {\bf 77}, 062337 (2008).

\bibitem[{\citenamefont{Kay}(2007)}]{kay-2006b}
\bibinfo{author}{\bibfnamefont{A.}~\bibnamefont{Kay}}, \bibinfo{journal}{Phys.~Rev.~Lett.} \textbf{\bibinfo{volume}{98}}, \bibinfo{pages}{010501}
  (\bibinfo{year}{2007}).

\bibitem{paternostro:2}  C.~Di Franco, M.~Paternostro, M.~S.~Kim, Phys.~Rev.~Lett.~{\bf 101}, 230502 (2008).

\bibitem{Kay:08b} A.~Kay, Phys.~Rev.~A {\bf 79}, 042330 (2009).

\bibitem{bruder} A.~O.~Lyakhov, C.~Bruder, Phys.~Rev.~B {\bf 74}, 235303 (2006).

\bibitem[{\citenamefont{Kay and Ericsson}(2005)}]{Kay:2005b}
\bibinfo{author}{\bibfnamefont{A.}~\bibnamefont{Kay}} \bibnamefont{and}
  \bibinfo{author}{\bibfnamefont{M.}~\bibnamefont{Ericsson}},
  \bibinfo{journal}{New J. Phys.} \textbf{\bibinfo{volume}{7}},
  \bibinfo{pages}{143} (\bibinfo{year}{2005}).


\bibitem[{\citenamefont{Shi et~al.}(2005)\citenamefont{Shi, Li, Song, and
  Sun}}]{shi}
\bibinfo{author}{\bibfnamefont{T.}~\bibnamefont{Shi}},
  \bibinfo{author}{\bibfnamefont{Y.}~\bibnamefont{Li}},
  \bibinfo{author}{\bibfnamefont{Z.}~\bibnamefont{Song}}, \bibnamefont{and}
  \bibinfo{author}{\bibfnamefont{C.~P.} \bibnamefont{Sun}},
  \bibinfo{journal}{Phys. Rev. A} \textbf{\bibinfo{volume}{71}},
  \bibinfo{pages}{032309} (\bibinfo{year}{2005}).

\bibitem{gladwell}
G.~M.~L.~Gladwell, Inverse Problems in Vibration, Kluwer
Academic, Boston (1986).

\bibitem[{\citenamefont{Karbach and Stolze}(2005)}]{transfer_comment}
\bibinfo{author}{\bibfnamefont{P.}~\bibnamefont{Karbach}} \bibnamefont{and}
  \bibinfo{author}{\bibfnamefont{J.}~\bibnamefont{Stolze}},
  \bibinfo{journal}{Phys. Rev. A} \textbf{\bibinfo{volume}{72}},
  \bibinfo{pages}{030301(R)} (\bibinfo{year}{2005}).

\bibitem{hochstadt} H.~Hochstadt, Lin.~Algebra Appl.~{\bf 8}, 435 (1974).

\bibitem[{\citenamefont{Albanese et~al.}(2004)\citenamefont{Albanese,
  Christandl, Datta, and Ekert}}]{Christandl:2004a}
\bibinfo{author}{\bibfnamefont{C.}~\bibnamefont{Albanese}},
  \bibinfo{author}{\bibfnamefont{M.}~\bibnamefont{Christandl}},
  \bibinfo{author}{\bibfnamefont{N.}~\bibnamefont{Datta}}, \bibnamefont{and}
  \bibinfo{author}{\bibfnamefont{A.}~\bibnamefont{Ekert}},
  \bibinfo{journal}{Phys. Rev. Lett} \textbf{\bibinfo{volume}{93}},
  \bibinfo{pages}{230502} (\bibinfo{year}{2004}).

\bibitem{xi} X.~Q.~Xi, J.~B.~Gong, T.~Zhang, R.~H.~Yue, and W.~M.~Liu, Eur.~Phys.~J.~D {\bf 50}, 193 (2008).

\bibitem[{\citenamefont{Cook and Shore}(1979)}]{shore}
\bibinfo{author}{\bibfnamefont{R.~J.} \bibnamefont{Cook}} \bibnamefont{and}
  \bibinfo{author}{\bibfnamefont{B.~W.} \bibnamefont{Shore}},
  \bibinfo{journal}{Phys. Rev. A} \textbf{\bibinfo{volume}{20}},
  \bibinfo{pages}{539} (\bibinfo{year}{1979}).

\bibitem[{\citenamefont{Kay}(2006{\natexlab{a}})}]{Kay:2005e}
\bibinfo{author}{\bibfnamefont{A.}~\bibnamefont{Kay}}, \bibinfo{journal}{Phys.
  Rev. A} \textbf{\bibinfo{volume}{73}}, \bibinfo{pages}{032306}
  (\bibinfo{year}{2006}{\natexlab{a}}).

\bibitem[{\citenamefont{Yung}(2006)}]{yung:06}
\bibinfo{author}{\bibfnamefont{M.-H.} \bibnamefont{Yung}},
  \bibinfo{journal}{Phys. Rev. A} \textbf{\bibinfo{volume}{74}},
  \bibinfo{pages}{030303(R)} (\bibinfo{year}{2006}).


\bibitem[{\citenamefont{Press et~al.}(2001)\citenamefont{Press, Teukolsky,
  Vetterling, and Flannery}}]{numerical_recipes}
\bibinfo{author}{\bibfnamefont{W.~H.} \bibnamefont{Press}},
  \bibinfo{author}{\bibfnamefont{S.~A.} \bibnamefont{Teukolsky}},
  \bibinfo{author}{\bibfnamefont{W.~T.} \bibnamefont{Vetterling}},
  \bibnamefont{and} \bibinfo{author}{\bibfnamefont{B.~P.}
  \bibnamefont{Flannery}}, \emph{\bibinfo{title}{Numerical Recipes in Fortran
  77}} (\bibinfo{publisher}{Cambridge University Press}, \bibinfo{year}{2001}).

\bibitem{marcin} M.~Wiesniak, Phys.~Rev.~A {\bf 78}, 052334 (2008).

\bibitem{innocent} S.~L'Innocente, C.~Lupo, S.~Mancini, J.~Phys.~A {\bf 42}, 475305 (2009).

\bibitem[{\citenamefont{Jordan and Wigner}(1928)}]{JordanWigner}
\bibinfo{author}{\bibfnamefont{P.}~\bibnamefont{Jordan}} \bibnamefont{and}
  \bibinfo{author}{\bibfnamefont{E.}~\bibnamefont{Wigner}},
  \bibinfo{journal}{Z. Phys.} \textbf{\bibinfo{volume}{47}},
  \bibinfo{pages}{631} (\bibinfo{year}{1928}).

\bibitem[{\citenamefont{Osborne}(2003)}]{Osborne}
\bibinfo{author}{\bibfnamefont{T.~J.} \bibnamefont{Osborne}}
  (\bibinfo{year}{2003}), \bibinfo{note}{quant-ph/0312126}.

\bibitem[{\citenamefont{Clark et~al.}(2005)\citenamefont{Clark, Moura-Alves,
  and Jaksch}}]{Jaksch:2004a}
\bibinfo{author}{\bibfnamefont{S.}~\bibnamefont{Clark}},
  \bibinfo{author}{\bibfnamefont{C.}~\bibnamefont{Moura-Alves}},
  \bibnamefont{and} \bibinfo{author}{\bibfnamefont{D.}~\bibnamefont{Jaksch}},
  \bibinfo{journal}{New J.~Phys.} \textbf{\bibinfo{volume}{7}},
  \bibinfo{pages}{124} (\bibinfo{year}{2005}).



\bibitem{marg} N.~Margolus and L.~B.~Levitin, Physica D {\bf 120}, 188
(1998).

\bibitem[{\citenamefont{Lloyd et~al.}(2003)\citenamefont{Lloyd, Landahl, and
  Slotine}}]{UQI}
\bibinfo{author}{\bibfnamefont{S.}~\bibnamefont{Lloyd}},
  \bibinfo{author}{\bibfnamefont{A.~J.}~\bibnamefont{Landahl}},
  \bibnamefont{and} \bibinfo{author}{\bibfnamefont{J.~E.}~\bibnamefont{Slotine}}, Phys.~Rev.~A {\bf 69}, 012305 (\bibinfo{year}{2004}).

\bibitem{Burgarth:07} D.~Burgarth and V.~Giovannetti (2007), arXiv: 0710.0302.

\bibitem{sgs} S.~G.~Schirmer, I.~C.~H.~Pullen and P.~J.~Pemberton-Ross Phys.~Rev.~A {\bf 78}, 062339 (2008).

\bibitem{Kay:08} A.~Kay, Phys.~Rev.~A {\bf 78} 012346 (2008).

\bibitem{Kay:09uqi} A.~Kay and P.~J.~Pemberton-Ross, Phys.~Rev.~A {\bf 81}, 010301(R) (2010).

\bibitem{Burgarth:uqi} D.~Burgarth, K.~Maruyama, M.~Murphy, S.~Montangero, T.~Calarco, F.~Nori and M.~B.~Plenio, arXiv:0905.3373 (2009).

\bibitem{lieb} E.~Lieb, T.~Schultz and D.~Mattis, Annals of Phys.~{\bf 16}, 407 (1961).

\bibitem{anderson} P.~W.~Anderson, Phys.~Rev.~{\bf 109}, 1492 (1958).

\bibitem{simone} G.~De Chiara, D.~Rossini, S.~Montangero, R.~Fazio, Phys.~Rev.~A {\bf 72}, 012323 (2005).

\bibitem{Burgarth} D.~Burgarth and S.~Bose, Phys.~Rev.~A {\bf 73}, 062321 (2006).

\bibitem{Cai} J.~M.~Cai, Z.~W.~Zhou and G.~C.~Guo, Phys.~Rev.~A {\bf 74}, 022328 (2006).

\bibitem[{\citenamefont{Nielsen and Chuang}(2000)}]{nielsen}
\bibinfo{author}{\bibfnamefont{M.~A.} \bibnamefont{Nielsen}} \bibnamefont{and}
  \bibinfo{author}{\bibfnamefont{I.~L.} \bibnamefont{Chuang}},
  \emph{\bibinfo{title}{Quantum Computation and Quantum Information}}
  (\bibinfo{publisher}{Cambridge University Press},
  \bibinfo{address}{Cambridge, UK}, \bibinfo{year}{2000}).

\bibitem{plenio:bio} F.~Caruso, A.~W.~Chin, A.~Datta, S.~F.~Huelga, M.~B.~Plenio,  J.~Chem.~Phys.~{\bf 131}, 105106 (2009).

\bibitem{Kay:square} A.~Kay, quant-ph/0702088 (2007).

\bibitem{Kay:09} A.~Kay, Phys.~Rev.~Lett.~{\bf 102}, 070503 (2009); F.~Pastawski, A.~Kay, N.~Schuch and I.~Cirac, Quant.~Inf.~Comp.~{\em to appear} (2010), arXiv:0911.3843.

\bibitem[{\citenamefont{Perales and Plenio}(2005)}]{plenio:05}
\bibinfo{author}{\bibfnamefont{A.}~\bibnamefont{Perales}} \bibnamefont{and}
  \bibinfo{author}{\bibfnamefont{M.~B.} \bibnamefont{Plenio}},
  \bibinfo{journal}{J. Opt. B: Quantum Semiclass. Opt.}
  \textbf{\bibinfo{volume}{7}}, \bibinfo{pages}{S601} (\bibinfo{year}{2005}).

\bibitem[{\citenamefont{Yang et~al.}(2005)\citenamefont{Yang, Song, and
  Sun}}]{yang:05}
\bibinfo{author}{\bibfnamefont{S.}~\bibnamefont{Yang}},
  \bibinfo{author}{\bibfnamefont{Z.}~\bibnamefont{Song}}, \bibnamefont{and}
  \bibinfo{author}{\bibfnamefont{C.~P.} \bibnamefont{Sun}},
  Eur.~Phys.~J.~B {\bf 52} 377 (2006).

\bibitem{damico}
I.~D'Amico, B.~W.~Lovett and T.~P.~Spiller, Phys. Rev. A {\bf 76}, 030302(R) (2007); I.~D'Amico, B.~W.~Lovett and T.~P.~Spiller, physica status solidi {\bf 7}, 2481 (2008).


\bibitem{kwek} Li Dai, Y.~P.~Feng and L.~C.~Kwek, J.~Phys.~A {\bf 43} 035302 (2010).

\bibitem{peres}
A.~Peres, Phys.~Rev.~A {\bf 32}, 3266 (1985).

\bibitem{feynman}
R.~P.~Feynman, Opt.~News {\bf 11}, 11 (1985).

\bibitem{karl} K.~G.~H.~Vollbrecht and J.I.~Cirac, Phys.~Rev.~Lett.~{\bf 100}, 010501 (2008).

\bibitem{Kay:2009x} A.~Kay, Phys.~Rev.~A {\bf 80}, 040301(R) (2009).

\bibitem{gottesman}
\bibinfo{author}{\bibfnamefont{D.}~\bibnamefont{Aharonov}} {\em et al.} Proceedings of the 48th IEEE Symposium on Foundations of Computer Science, FOCS '07 (IEEE, Providence, RI, 2007),
p. 373.

\bibitem[{\citenamefont{Kitaev et~al.}(2002)\citenamefont{Kitaev, Shen, and
  Vyalyi}}]{KSV02a}
\bibinfo{author}{\bibfnamefont{A.~Y.}~\bibnamefont{Kitaev}},
  \bibinfo{author}{\bibfnamefont{A.~H.}~\bibnamefont{Shen}}, \bibnamefont{and}
  \bibinfo{author}{\bibfnamefont{M.~N.}~\bibnamefont{Vyalyi}},
  \emph{\bibinfo{title}{Classical and quantum computation}}, Graduate studies
  in mathematics (\bibinfo{publisher}{AMS},
  \bibinfo{address}{Providence, Rhodes Island}, \bibinfo{year}{2002}).

\bibitem[{\citenamefont{Oliveira and Terhal}(2005)}]{oliveira-2005}
\bibinfo{author}{\bibfnamefont{R.}~\bibnamefont{Oliveira}} \bibnamefont{and}
  \bibinfo{author}{\bibfnamefont{B.~M.}~\bibnamefont{Terhal}}, Quant.~Inf.~Comp.~{\bf 8}, 0900 (\bibinfo{year}{2008}).


\bibitem[{\citenamefont{Kempe et~al.}(2006)\citenamefont{Kempe, Kitaev, and
  Regev}}]{Kempe}
\bibinfo{author}{\bibfnamefont{J.}~\bibnamefont{Kempe}},
  \bibinfo{author}{\bibfnamefont{A.}~\bibnamefont{Kitaev}}, \bibnamefont{and}
  \bibinfo{author}{\bibfnamefont{O.}~\bibnamefont{Regev}},
  \bibinfo{journal}{SIAM Journal of Computing} \textbf{\bibinfo{volume}{35}},
  \bibinfo{pages}{1070} (\bibinfo{year}{2006}).

\bibitem{irani} S. Irani, J.~Math.~Phys.~{\bf 51}, 022101 (2010).
\end{thebibliography}
\end{document}